\documentclass[12pt,a4paper]{article}

\usepackage[T1]{fontenc}
\usepackage[utf8]{inputenc}

\usepackage{amsmath,amssymb,amsthm,mathtools}

\usepackage{bm}
\usepackage{aliascnt}

\usepackage{geometry}
\usepackage{xcolor}
\usepackage{hyperref}
\hypersetup{
  colorlinks=true,
  linkcolor=blue!60!black,
  citecolor=blue!60!black,
  urlcolor=blue!60!black
}

\usepackage[nameinlink,noabbrev]{cleveref}
\usepackage{enumitem}
\usepackage{booktabs}
\usepackage{array}
\usepackage{threeparttable}
\usepackage{etoolbox}   

\theoremstyle{plain}
\newtheorem{theorem}{Theorem}[section]

\newaliascnt{lemma}{theorem}
\newtheorem{lemma}[lemma]{Lemma}
\aliascntresetthe{lemma}

\newaliascnt{proposition}{theorem}
\newtheorem{proposition}[proposition]{Proposition}
\aliascntresetthe{proposition}

\newaliascnt{corollary}{theorem}
\newtheorem{corollary}[corollary]{Corollary}
\aliascntresetthe{corollary}

\theoremstyle{definition}
\newaliascnt{definition}{theorem}
\newtheorem{definition}[definition]{Definition}
\aliascntresetthe{definition}

\newaliascnt{example}{theorem}
\newtheorem{example}[example]{Example}
\aliascntresetthe{example}

\theoremstyle{remark}
\newaliascnt{remark}{theorem}
\newtheorem{remark}[remark]{Remark}
\aliascntresetthe{remark}

\crefname{theorem}{theorem}{theorems}
\Crefname{theorem}{Theorem}{Theorems}
\crefname{lemma}{lemma}{lemmas}
\Crefname{lemma}{Lemma}{Lemmas}
\crefname{proposition}{proposition}{propositions}
\Crefname{proposition}{Proposition}{Propositions}
\crefname{corollary}{corollary}{corollaries}
\Crefname{corollary}{Corollary}{Corollaries}
\crefname{definition}{definition}{definitions}
\Crefname{definition}{Definition}{Definitions}
\crefname{example}{example}{examples}
\Crefname{example}{Example}{Examples}
\crefname{remark}{remark}{remarks}
\Crefname{remark}{Remark}{Remarks}
\crefname{appendix}{Appendix}{appendix}
\Crefname{appendix}{Appendix}{Appendices}

\pretocmd{\appendix}{%
  \crefalias{section}{appendix}%
  \crefalias{subsection}{appendix}%
  \crefalias{subsubsection}{appendix}%
}{}{}

\newcommand{\R}{\mathbb{R}}

\newcommand{\bra}[1]{\langle#1\rvert}
\newcommand{\ket}[1]{\lvert#1\rangle}
\newcommand{\braket}[2]{\langle#1\vert#2\rangle}
\newcommand{\Ry}{R_Y}
\newcommand{\Gclass}[2]{\mathcal{G}_{#1}^{(#2)}}
\newcommand{\Thetag}{\Theta_g}

\newcommand{\wh}[2]{\hat{#1}_{#2}}

\newcommand{\cnot}{\mathrm{CNOT}}

\newcommand{\Iver}[1]{\mathbf{1}_{#1}}
\newcommand{\chimon}[1]{\chi_{#1}}
\newcommand{\indvec}[1]{\mathbf{1}_{#1}}

\newenvironment{keywords}{%
  \par\medskip\noindent\textbf{Keywords.}\enspace\ignorespaces
}{%
  \par\medskip
}
\newenvironment{subjclass}[1][2020]{%
  \par\noindent\textbf{Mathematics Subject Classification (#1).}\enspace\ignorespaces
}{%
  \par\bigskip
}

\title{%
  \textbf{On the encoding complexity of quantum numerical integration:
  an angle-structure characterization}\\[6pt]
  }

\author{%
  F. Chinesta\thanks{PIMM Lab, Arts et Métiers ParisTech, Paris, France. Email: \texttt{francisco.chinesta@ensam.eu}}
  \and
  A. Falcó\thanks{Corresponding author. Departamento de Matemáticas, Física y Ciencias Tecnológicas, Universidad CEU Cardenal Herrera, 
  Alfara del Patriarca, Valencia, Spain. Email: \texttt{afalco@uchceu.es}}
  \and
  D. Falcó--Pomares\thanks{Grupo de Investigación 
  Bisite, Universidad de Salamanca, Calle Espejo s/n, 37007 Salamanca, Spain. 
  Email: \texttt{dfp99@usal.es}}}

\date{\today}

\begin{document}
\maketitle

\begin{abstract}
We study numerical integration on $[0,1]$ by quantum amplitude
estimation (QAE), with emphasis on the cost of constructing the
amplitude oracle. While QAE improves the statistical component of
the integration error, this gain is relevant to the full resource
count only when the integrand admits low encoding complexity.

We introduce a hierarchy of function classes $\mathcal{G}_n^{(d)}$
on the $n$-qubit grid, defined by requiring the angle map
$\Theta_g : \{0,1\}^n \to [0,\pi]$ to be multilinear of degree at
most $d$. Membership is classically checkable in $O(n2^n)$ time via
the Walsh--Hadamard transform (WHT). For $g \in \mathcal{G}_n^{(d)}$,
the encoding operator factorises into
$\sum_{k=0}^d \binom{n}{k}$ multi-controlled $R_Y$ gates
(an upper bound on gate complexity, tight at $d=n$),
interpolating between the affine regime $O(n)$ and the generic
exponential regime.

Combining this with classical discretisation estimates for
$g \in C^\alpha[0,1]$, we derive a depth-versus-accuracy trade-off:
total gate count $O((\log(1/\varepsilon))^d \cdot \varepsilon^{-1})$
suffices to achieve $\varepsilon$-accuracy with constant probability.
For $d=1$ this gives $O(\varepsilon^{-1}\log(1/\varepsilon))$,
below the corresponding classical Monte Carlo sample count at fixed
discretisation for every $\alpha \ge 1$.
We further show that the encoding degree is not determined by the
regularity of the integrand: for every $s \in (0,1/2)$ the class
$\mathcal{G}_n^{(1)}$ contains restrictions of functions
$g_s : [0,1] \to [0,1]$ that belong to $W^{s',2}(0,1)$ for every
$s' < s$ but not to $W^{s,2}(0,1)$.
Encoding cost and Sobolev smoothness are therefore independent
parameters of the problem.

These results identify a mathematically explicit class of integrands
for which the state-preparation cost of QAE-based integration remains
polylogarithmic in the target accuracy.
Hardware experiments on the SpinQ Triangulum (NMR) and IBM Kingston
(superconducting) processors validate the hierarchy at $n=2$:
circuits in $\mathcal{G}_n^{(1)}$ execute reliably on both platforms,
whereas $\mathcal{G}_n^{(2)}$ circuits exceed the Triangulum coherence
budget and decohere completely, confirming the depth bound as a
physical threshold.
\end{abstract}


\begin{subjclass}[2020]
65D30, 65Y20, 68Q12, 81P68, 41A25, 42C10, 46E35
\end{subjclass}

\section{Introduction}
\label{sec:intro}

The complexity of numerical integration is classically governed by the
interplay between regularity and sampling cost.
For $g \in C^\alpha[0,1]$, an $N$-point quadrature rule of order~$p\leq\alpha$
satisfies $|I - Q_N^{(p)}[g]| = \mathcal{O}(N^{-p})$
\cite{DavisRabinowitz1984}, while Monte Carlo estimation contributes an
additional statistical error of order $\mathcal{O}(M^{-1/2})$ using
$M$ random samples.
The total cost of achieving $\varepsilon$-accuracy therefore reflects
two distinct mechanisms: a discretisation error controlled by the
smoothness of the integrand, and a statistical error controlled by the
number of function evaluations.

Quantum Amplitude Estimation (QAE) \cite{BrassardHoyerMoscaTapp2002}
improves the statistical component of this balance by replacing the
classical Monte Carlo rate $\mathcal{O}(M^{-1/2})$ with the quantum rate
$\mathcal{O}(M^{-1})$ after $M$ oracle calls.
This observation has motivated a growing literature on quantum numerical
integration \cite{AbramsWilliams1999,Montanaro2015} and quantum finance
\cite{WoernerEgger2019,StamatopoulosEtAl2019,EggerEtAl2019}.
From the perspective of numerical analysis, however, the relevant
question is not whether QAE improves sampling complexity in isolation,
but whether this sampling gain remains visible once the cost of
constructing the amplitude oracle is incorporated into the resource
count.

This issue is decisive.
The QAE gain applies to the \emph{estimation} stage, not
automatically to the \emph{encoding} stage.
If the circuit that loads~$g$ into amplitudes has depth~$D$, then the
cost of one QAE run is $D + \mathcal{O}(k)$, where $k$ is the
amplification level.
For sufficiently large~$k$, the gain in the estimation stage may
dominate; but if $D$ grows exponentially with the number of qubits~$n$,
then the resource count is dominated by state preparation.
Indeed, preparing a generic $n$-qubit state requires $\Theta(2^n)$
gates \cite{ShendeBullockMarkov2006,BarencoEtAl1995}.
The work of Carrera Vazquez \& Woerner
\cite{CarreraVazquezWoerner2020} identified practically relevant
function classes admitting efficient encodings and showed how a
Spin-Echo optimisation reduces the depth of repeated amplitude
amplification.
What remained open is the structural question that is central for the
complexity of quantum numerical integration:
\emph{which integrands admit low-complexity encodings, and how can they be
characterised mathematically?}

The state-preparation bottleneck is therefore not merely a technical
detail; it is a structural component of the complexity of the full
numerical method.
A recent survey of quantum alternatives to classical Monte Carlo
\cite{IntalluraKorpas2026} identifies three principal obstacles in
practical QAE implementations: the construction of the amplitude oracle,
the depth of the QAE circuit, and the complexity of state preparation.
It also notes explicitly that the Grover--Rudolph procedure
\cite{GroverRudolph2002} may eliminate the quadratic speedup entirely
for log-concave distributions \cite{Herbert2021}.
At the same time, the survey establishes, through a systematic complexity
comparison for H\"older and Sobolev classes, that quantum query
complexity is asymptotically better than the classical randomised bound
across all regularity regimes, provided that the state-preparation cost
is negligible.

The present paper is complementary to the information-based complexity
analysis of quantum integration in Sobolev classes carried out by
Novak~\cite{Novak2001}, Heinrich~\cite{Heinrich2003}, and
Heinrich \& Novak~\cite{HeinrichNovak2003}.
Those works determine the optimal query rates in a model where the
amplitude oracle is a unit-cost black box; we leave those rates
untouched and address instead the structural question of which
integrands admit low-cost amplitude encodings.
Rather than treating state preparation as a black box, we characterise,
through the angle-structure hierarchy $\{\mathcal{G}_n^{(d)}\}$, the
subclass of integrands for which the encoding operator remains sufficiently simple
for the improved sampling behaviour of QAE not to be dominated by
state-preparation cost.

The distinction from earlier work is precise.
Carrera Vazquez \& Woerner~\cite{CarreraVazquezWoerner2020} identified
a practically relevant subclass of functions admitting low-depth
encodings via a Spin-Echo optimisation, and demonstrated reduced circuit
depth for repeated amplitude amplification; their analysis, however,
takes the encoding cost as given and does not provide a mathematical
characterisation of which integrands admit such encodings or why.
Suzuki et al.~\cite{SuzukiEtAl2019} established the statistical
convergence properties of the MLAE estimator without addressing the
encoding cost at all.
Heinrich~\cite{Heinrich2002} studied the quantum query complexity of
summation and integration in $L_p([0,1]^d)$ within the
information-based complexity framework, establishing optimal convergence
rates in the oracle model; in that setting the amplitude oracle is
treated as a black box whose construction cost is not charged.
The present paper operates in a complementary regime: we charge the
circuit cost of building the oracle and ask which structural property of
the integrand determines that cost.
The answer --- the multilinear degree of the angle map --- is an
invariant that has no counterpart in the oracle model, and it is this
invariant that underlies the encoding-complexity theorem proved here.
We emphasise that our results are statements about circuit cost within
the oracle-construction stage; we make no claim about the
information-based complexity of integration, for which the
deterministic, randomised and quantum rates over Sobolev balls are
already known \cite{Novak2001,Heinrich2002,Heinrich2003}.

Our starting point is the angle map
\[
\Theta_g : \{0,1\}^n \to [0,\pi], \qquad
\Theta_g(b)=2\arcsin\!\bigl(\sqrt{g(x_{i(b)})}\bigr),
\]
where $b=(b_0,\ldots,b_{n-1})\in\{0,1\}^n$,
$i(b) = \sum_{k=0}^{n-1} b_k\,2^k \in \{0,\ldots,2^n-1\}$
is the associated grid index, and
\[
  \mathcal{X}_n := \bigl\{x_{i(b)} = i(b)/2^n : b \in \{0,1\}^n\bigr\}
\]
is the uniform $n$-qubit grid.
The angle map $\Theta_g$ determines the controlled rotations appearing
in the amplitude oracle on~$\mathcal{X}_n$.
We define $\mathcal{G}_n^{(d)}$ as the class of functions
$g:\mathcal{X}_n\to[0,1]$ whose angle map is multilinear of degree at
most~$d$.
This degree is classically computable in $\mathcal{O}(n2^n)$ time via
the Walsh--Hadamard transform and provides an exact structural
description of the corresponding encoding complexity.
For $g|_{\mathcal{X}_n} \in \mathcal{G}_n^{(d)}$, we prove that the
encoding operator factorises exactly into $\sum_{k=0}^d \binom{n}{k}$
multi-controlled $\Ry$ gates, thereby yielding a hierarchy that
interpolates between the affine regime $\mathcal{O}(n)$ and the generic
exponential regime.

Combining this structural result with classical discretisation bounds,
we derive a depth-versus-accuracy trade-off for QAE-based numerical
integration.
In particular, for $g:[0,1]\to[0,1]$ with $g|_{\mathcal{X}_n} \in
\mathcal{G}_n^{(d)}$ and $g \in C^\alpha[0,1]$,
$\varepsilon$-accuracy is achieved with total gate count
\[
\mathcal{O}\!\left(\frac{(\log(1/\varepsilon))^d}{\varepsilon}\right),
\]
and in the affine case $d=1$ this reduces to
$\mathcal{O}(\varepsilon^{-1}\log(1/\varepsilon))$.
We further show that the degree of the angle map is not controlled by
the smoothness of the integrand: $\mathcal{G}_n^{(1)}$ contains
restrictions of functions of arbitrarily low Sobolev regularity, whose
encoding depth is nevertheless $\mathcal{O}(n)$.
The oscillation responsible for low regularity lives \emph{within}
the grid cells, whereas the degree of the angle map depends only on the
cell boundary values; the two are therefore independent.
In this way, the paper identifies a mathematically explicit class of
integrands for which the state-preparation stage of QAE does not
dominate the overall cost.

\paragraph{Main contributions.}
The paper develops a structural complexity theory for the amplitude
encoding step in QAE-based numerical integration.
Its main contributions are as follows.

\begin{enumerate}[label=(\roman*)]
\item
  \textbf{Angle-structure hierarchy} (\cref{sec:hierarchy}).
  We introduce the class $\Gclass{n}{d}$ of functions
  $g:\mathcal{X}_n\to[0,1]$ whose angle map
  $\Thetag:\{0,1\}^n\to[0,\pi]$,
  $\Thetag(b)=2\arcsin(\sqrt{g(x_{i(b)})})$,
  is a multilinear polynomial of degree at most~$d$.
  Membership in $\Gclass{n}{d}$ is checkable in
  $\mathcal{O}(n2^n)$ time via the Walsh--Hadamard transform, and we
  give an intrinsic characterisation of $\Gclass{n}{1}$ as a finite
  system of linear constraints on the values of~$g$.
  We also show that the hierarchy is compatible with the product and
  sum constructions of \cite{CarreraVazquezWoerner2020}: neither
  operation increases the degree of the angle map.

\item
  \textbf{Encoding-circuit theorem} (\cref{sec:hierarchy}).
  For $g\in\Gclass{n}{d}$, the encoding operator admits an exact
  factorisation into $\binom{n}{\leq d}$ multi-controlled $\Ry$ gates,
  one for each multilinear monomial of degree at most~$d$.
  This gives an upper bound on the gate complexity that interpolates
  between $n+1$ in the affine case ($d=1$) and $2^n$ in the generic
  case ($d=n$); the endpoint $d=n$ coincides with the general lower
  bound for arbitrary state preparation \cite{ShendeBullockMarkov2006},
  confirming that the monomial factorisation is optimal when no
  degree constraint is available.
  Whether $\binom{n}{\leq d}$ is tight for $d<n$ is identified as an
  open problem (see \S\ref{sec:conclusions}).

\item
  \textbf{Depth-versus-accuracy trade-off} (\cref{sec:pareto}).
  Combining the gate-count estimate with classical discretisation-error
  bounds, we derive the joint scaling of encoding depth and integration
  accuracy as a function of $n$, $d$, and the quadrature order~$p$.
  The total gate count interpolates between
  $\mathcal{O}(\varepsilon^{-1}\log(1/\varepsilon))$ for $d=1$ and
  $\mathcal{O}(\varepsilon^{-(1+1/p)})$ for $d=n^*$, with a
  polylogarithmic factor $(\log(1/\varepsilon))^d$ at intermediate
  degrees.
  In particular, for $g\in\Gclass{n}{1}$, the total gate count is
  $\mathcal{O}(\varepsilon^{-1}\log(1/\varepsilon))$; under the
  fixed-discretisation comparison made precise in \cref{rem:scope},
  this is below the corresponding classical Monte Carlo sample count
  for every $p\ge 1$.

\item
  \textbf{Independence of encoding degree and regularity}
  (\cref{sec:pareto}).
  We construct, for every $s\in(0,1/2)$, a family of functions
  $g_s:[0,1]\to[0,1]$ with $g_s|_{\mathcal{X}_n}\in\Gclass{n}{1}$ that
  lie in $W^{s',2}(0,1)$ for all $s'<s$ but not in $W^{s,2}(0,1)$.
  The encoding depth of $g_s$ is $\mathcal{O}(n)$ irrespective
  of~$s$, showing that the angle-structure degree is not a smoothness
  invariant and cannot be read off from the regularity of the
  integrand.

\item
  \textbf{Computational illustration} (\cref{sec:experiments}).
  We implement the QAE protocol on two hardware platforms: the SpinQ
  Triangulum 3-qubit NMR processor \cite{SpinQ2021} and the IBM Kingston
  127-qubit superconducting processor \cite{IBMQuantum2024}, for functions
  of degrees $d=0$, $d=1$, and $d=2$ at $n=2$.
  The experiments illustrate that the $\Gclass{2}{1}$ encoding
  (3 gates, $\mathcal{O}(n)$ depth) executes stably on both platforms,
  whereas the $\Gclass{2}{2}$ encoding (4 gates) exceeds the Triangulum
  hardware depth limit but executes successfully on Kingston.
  This shows that the $\Gclass{n}{d}$
  degree stratum acts as a practical hardware-feasibility criterion,
  consistent with the gate-count predictions of \cref{thm:circuit_depth}
  at $n=2$.
\end{enumerate}

\paragraph{Organisation.}
\Cref{sec:problem} formulates the numerical integration problem within
the QAE framework.
\Cref{subsec:qprob} introduces quantum probability on a finite index set
in standard linear-algebra notation, while
\Cref{subsec:oracle,subsec:mlae} define the amplitude oracle and the
MLAE estimator.
\Cref{sec:encoding} introduces the encoding map and its multilinear
structure.
\Cref{sec:hierarchy} develops the angle-structure hierarchy together
with the corresponding encoding-complexity theorem.
\Cref{sec:pareto} contains the main complexity and error-analysis
results, together with the construction decoupling encoding degree
from Sobolev regularity.
\Cref{sec:experiments} provides a computational illustration of the
hierarchy at $n=2$.
\Cref{sec:conclusions} concludes with open problems, including
multivariate extensions and a possible application to the Heston model.
\Cref{app:qc} presents the quantum-computing model in matrix language,
together with a dictionary relating it to Dirac notation.
\Cref{app:circuit_proof,app:walsh} collect the circuit implementation
details and Walsh--Hadamard background.

\section{Problem Formulation}
\label{sec:problem}

\subsection{The numerical integration problem}

Let $g:[0,1]\to[0,1]$ be a bounded measurable function.
We wish to compute the integral $I[g] := \int_0^1 g(x)\,dx$ to
precision~$\varepsilon > 0$.

Fix $n \geq 1$ and let $N := 2^n$.
The uniform $n$-qubit grid is
\begin{equation}
  \mathcal{X}_n := \bigl\{x_i = i/N : i \in \{0,1,\ldots,N-1\}\bigr\},
  \label{eq:grid}
\end{equation}
Let $p \geq 1$ be an integer and let $Q_n^{(p)}[g]$ denote the approximation
of $I[g]$ produced by a quadrature rule of order~$p$ on the grid
$\mathcal{X}_n$.
For $g \in C^p[0,1]$ the discretisation error satisfies
\begin{equation}
  E_n^{(p)}[g]
  := \bigl|I[g] - Q_n^{(p)}[g]\bigr|
  = \mathcal{O}(2^{-pn}),
  \label{eq:disc_error}
\end{equation}
and more generally, for $g \in C^\alpha[0,1]$ with $\alpha \geq p$,
the same bound $\mathcal{O}(2^{-pn})$ holds since the rule exploits
at most $p$ derivatives.
Explicit bounds for the three standard rules ($p=1,2,4$) are recalled in
\cref{sec:pareto}.

\subsection{Quantum probability on a finite index set}
\label{subsec:qprob}

We collect the probabilistic formalism used throughout the paper.
The presentation uses standard linear-algebra notation;
a dictionary with the Dirac bra-ket notation common in the physics
literature is given in \cref{app:dirac}.

\paragraph{State space.}
Fix $N = 2^n$ and let $V = \mathbb{C}^N$ be equipped with the standard
Hermitian inner product $\langle \mathbf{u}, \mathbf{v}\rangle
= \mathbf{u}^* \mathbf{v}$.
Let $\{\mathbf{u}_0, \mathbf{u}_1, \ldots, \mathbf{u}_{N-1}\}$ be the
standard orthonormal basis of~$V$, indexed by $\mathbb{Z}_N =
\{0,1,\ldots,N-1\}$.
A \emph{quantum state} (or \emph{state}) is a unit vector $\mathbf{v} \in V$,
$\|\mathbf{v}\| = 1$.

\paragraph{Quantum random variable and projectors.}
A \emph{quantum random variable} over $\mathbb{Z}_N$ associated to the
basis $\{\mathbf{u}_k\}$ is the family of rank-one orthogonal projectors
\begin{equation}
  \Pi_k := \mathbf{u}_k \mathbf{u}_k^* \in \mathbb{M}_N(\mathbb{C}),
  \qquad k \in \mathbb{Z}_N.
  \label{eq:projectors}
\end{equation}
These satisfy $\Pi_k \Pi_j = \delta_{kj}\Pi_k$ (orthogonality),
$\Pi_k^* = \Pi_k$ (self-adjointness), and
$\sum_{k=0}^{N-1}\Pi_k = I_N$ (completeness).
The family $\{\Pi_k\}_{k=0}^{N-1}$ is a \emph{projection-valued measure}
(PVM) on $\mathbb{Z}_N$.

\paragraph{Probability law and expectation.}
Given a state $\mathbf{v} \in V$, a measurement of the quantum random
variable $\{\Pi_k\}$ yields outcome $k \in \mathbb{Z}_N$ with
probability
\begin{equation}
  \mathbb{P}(k;\, \mathbf{v})
  := \mathbf{v}^* \Pi_k\, \mathbf{v}
   = |\langle \mathbf{u}_k, \mathbf{v}\rangle|^2.
  \label{eq:born_linalg}
\end{equation}
One verifies $\mathbb{P}(k;\,\mathbf{v}) \geq 0$ and
$\sum_k \mathbb{P}(k;\,\mathbf{v}) = \mathbf{v}^* (\sum_k \Pi_k)\mathbf{v}
= \|\mathbf{v}\|^2 = 1$.
The \emph{expectation} of a function $f:\mathbb{Z}_N\to\mathbb{R}$
in state $\mathbf{v}$ is
\begin{equation}
  \mathbb{E}_{\mathbf{v}}[f]
  := \sum_{k=0}^{N-1} f(k)\,\mathbb{P}(k;\,\mathbf{v})
   = \mathbf{v}^* \!\Bigl(\sum_{k=0}^{N-1} f(k)\,\Pi_k\Bigr)\mathbf{v}.
  \label{eq:expectation_linalg}
\end{equation}

\paragraph{Circuits as unitary maps.}
A \emph{quantum circuit} is a unitary $U \in \mathrm{U}(N)$ built as
a product of elementary gates (rotations and controlled operations;
see \cref{app:qc}).
Starting from the fixed initial state $\mathbf{e}_0 := \mathbf{u}_0
\in V$, the circuit produces the output state
$\mathbf{v}_U := U\mathbf{e}_0$,
and the measurement probability is
\begin{equation}
  \mathbb{P}(k;\, U)
  := \mathbb{P}(k;\, U\mathbf{e}_0)
   = (U\mathbf{e}_0)^*\Pi_k(U\mathbf{e}_0)
   = |\langle \mathbf{u}_k, U\mathbf{e}_0\rangle|^2.
  \label{eq:circuit_prob}
\end{equation}

\paragraph{Product spaces and marginal probabilities.}
The amplitude oracle acts on the \emph{product space}
$V \otimes \mathbb{C}^2 = \mathbb{C}^{2N}$,
with standard basis
$\{\mathbf{u}_i \otimes \mathbf{u}_b^{(2)} :
i \in \mathbb{Z}_N,\, b \in \{0,1\}\}$,
where $\{\mathbf{u}_0^{(2)}, \mathbf{u}_1^{(2)}\}$ is the standard
basis of the \emph{ancilla} factor $\mathbb{C}^2$.
The projector onto ancilla state $b=1$ is
\begin{equation}
  \Pi_1^{\mathrm{anc}}
  := I_N \otimes \mathbf{u}_1^{(2)}(\mathbf{u}_1^{(2)})^*
  \;\in\; \mathbb{M}_{2N}(\mathbb{C}),
  \label{eq:ancilla_projector}
\end{equation}
and the marginal probability of measuring the ancilla in state~$1$,
for a circuit $U \in \mathrm{U}(2N)$ starting from
$\mathbf{e}_0 \otimes \mathbf{u}_0^{(2)}$, is
\begin{equation}
  \mathbb{P}(\mathrm{anc}=1;\, U)
  := (U(\mathbf{e}_0\otimes\mathbf{u}_0^{(2)}))^*\,
     \Pi_1^{\mathrm{anc}}\,
     (U(\mathbf{e}_0\otimes\mathbf{u}_0^{(2)})).
  \label{eq:marginal_linalg}
\end{equation}
This is the central quantity estimated by QAE throughout the paper.

\begin{remark}[Connection to algebraic probability]
\label{rem:alg_prob}
The pair $(\mathbb{M}_{2N}(\mathbb{C}),\,\mathbf{v}\mapsto\mathbf{v}^*(\cdot)\mathbf{v})$
is an instance of the algebraic-probability framework of
\cite{FalcoMatthies2026}: the algebra is $\mathbb{M}_{2N}(\mathbb{C})$,
the state is the vector functional $\omega_\mathbf{v}(A)=\mathbf{v}^*A\mathbf{v}$,
and $\{\Pi_k\}$ is a maximal abelian family of projectors generating
the \emph{classical output context}.
All probabilistic identities used in this paper follow from the
algebraic structure alone, without reference to physical axioms.
\end{remark}

\subsection{The amplitude oracle}
\label{subsec:oracle}

The QAE approach encodes the left Riemann sum of~$g$ as the marginal
probability~\eqref{eq:marginal_linalg} of a specific circuit.
For $g:\mathcal{X}_n\to[0,1]$, define the \emph{encoding operator}
$G_g\in\mathrm{U}(2N)$ block-diagonally by its action on the basis of
$V\otimes\mathbb{C}^2$,
\begin{equation}
  G_g\,(\mathbf{u}_i\otimes\mathbf{u}_0^{(2)})
  = \mathbf{u}_i\otimes
    \bigl(\sqrt{1-g(x_i)}\,\mathbf{u}_0^{(2)}
         +\sqrt{g(x_i)}\,\mathbf{u}_1^{(2)}\bigr),
  \qquad i\in\{0,\ldots,N-1\},
  \label{eq:Gg_def}
\end{equation}
and let $H_n = H^{\otimes n}$ with
$H = \tfrac{1}{\sqrt{2}}
\bigl(\begin{smallmatrix}1&1\\1&-1\end{smallmatrix}\bigr)$,
so that $H_n\mathbf{e}_0 = N^{-1/2}\sum_{i=0}^{N-1}\mathbf{u}_i$.
The \emph{amplitude oracle} is
\begin{equation}
  A_g := G_g\,(H_n\otimes I_2) \;\in\; \mathrm{U}(2N),
  \label{eq:Ag_def}
\end{equation}
and a direct computation from \eqref{eq:marginal_linalg} shows that its
output state
$\mathbf{w}_g := A_g(\mathbf{e}_0\otimes\mathbf{u}_0^{(2)})$
satisfies
\begin{equation}
  \mathbb{P}(\mathrm{anc}=1;\, A_g)
  = \mathbf{w}_g^*\,\Pi_1^{\mathrm{anc}}\,\mathbf{w}_g
  = \frac{1}{N}\sum_{i=0}^{N-1} g(x_i)
  =: R_n^{\mathrm{left}}[g].
  \label{eq:oracle_prob}
\end{equation}
Thus $A_g$ acts on an index register ($V = \mathbb{C}^N$) and one
ancilla qubit ($\mathbb{C}^2$).
Its gate complexity is measured by the number of elementary
controlled-rotation factors in its decomposition.
The integer~$n$ (through $N=2^n$) controls the approximation quality,
while the structure of~$g$ controls the complexity of~$A_g$.

The \emph{Grover iterate} associated with $A_g$ is
\begin{equation}
  Q := A_g\,S_{\mathbf{e}_0}\,A_g^*\,S_{\mathbf{w}_0}
  \;\in\;\mathrm{U}(2N),
  \qquad S_{\mathbf{v}} := I - 2\mathbf{v}\mathbf{v}^*,
  \label{eq:grover_def}
\end{equation}
where $S_{\mathbf{v}}$ is the reflection through the hyperplane
orthogonal to~$\mathbf{v}$ and $\mathbf{w}_0$ is the projection of
$\mathbf{w}_g$ onto the ancilla-$0$ component.
Its key property is
\begin{equation}
  \mathbb{P}(\mathrm{anc}=1;\, Q^k A_g)
  = \sin^2\!\bigl((2k+1)\arcsin(\sqrt{a})\bigr),
  \qquad k \in \mathbb{N}_0,
  \label{eq:qae_prob}
\end{equation}
where $a = R_n^{\mathrm{left}}[g]$.
This oscillation in~$k$ encodes the left Riemann sum of~$g$ in the
frequency $\arcsin(\sqrt{a})/\pi$, and $M$ evaluations of the iterate
suffice to estimate~$a$ with mean-squared error $\mathcal{O}(M^{-1})$
\cite{BrassardHoyerMoscaTapp2002}.

\subsection{Maximum Likelihood Amplitude Estimation (MLAE)}
\label{subsec:mlae}

We use the circuit-efficient variant of \cite{SuzukiEtAl2019}, which
replaces quantum phase estimation by a classical maximum-likelihood step.

\begin{definition}[MLAE estimator]
\label{def:mlae}
Let $\mathcal{K} = \{k_0,\ldots,k_r\} \subset \mathbb{N}_0$ be a
finite schedule and $N_k \in \mathbb{N}$ shots per level.
For each $k\in\mathcal{K}$, run the circuit $Q^k A_g$ and record
$m_k$ outcomes with ancilla~$=1$ out of $N_k$ shots.
The \emph{log-likelihood function} is
\begin{equation}
  \ell(a) :=
  \sum_{k\in\mathcal{K}}
  \bigl[m_k \log p_k(a)
       + (N_k - m_k)\log(1-p_k(a))\bigr],
  \label{eq:loglik}
\end{equation}
where
\[
  p_k(a) := \sin^2\!\bigl((2k+1)\arcsin(\sqrt{a})\bigr),
\]
and the MLAE estimator is
$\hat{a} := \arg\max_{a\in[0,1]}\ell(a)$.
\end{definition}

The total oracle cost is $M := \sum_{k\in\mathcal{K}} N_k(2k+1)$.
For a linear schedule $k_j = j$ with equal shots $N_k = N_0$, the
MLAE estimator is asymptotically efficient: as the total shot count
grows, it converges to an unbiased estimator achieving the
Cram\'{e}r--Rao lower bound
\cite[Theorem~1 and Eq.~(13)]{SuzukiEtAl2019}, so that
\begin{equation}
  \mathbb{E}\bigl[(\hat{a}-a)^2\bigr]^{1/2}
  \;\leq\; \frac{C_{\mathrm{est}}}{M},
  \label{eq:mlae_mse}
\end{equation}
for a constant $C_{\mathrm{est}} > 0$ depending only on the schedule
and on $a(1-a)$.
By Chebyshev's inequality, for any $\delta\in(0,1)$,
\begin{equation}
  \mathbb{P}\!\Bigl(|\hat{a}-a| \leq \frac{C_{\mathrm{est}}}{M\sqrt{\delta}}\Bigr)
  \;\geq\; 1 - \delta.
  \label{eq:mlae_hp}
\end{equation}
In particular, taking $\delta$ a fixed constant (e.g.\ $\delta = 1/4$)
gives $|\hat{a}-a| \leq 2C_{\mathrm{est}}/M$ with probability at least
$3/4$, so $\varepsilon$-accuracy with constant success probability
requires $M = \mathcal{O}(1/\varepsilon)$ oracle calls, independent of
the regularity of~$g$.

The statistical precision of $\hat{a}$ is bounded below by the
Cram\'{e}r--Rao bound $\mathrm{Var}(\hat{a})\geq 1/\mathcal{I}$,
where the total Fisher information is
\begin{equation}
  \mathcal{I} := \sum_{k\in\mathcal{K}} N_k\,I_k(a),
  \qquad
  I_k(a) :=
  \frac{\bigl[p_k'(a)\bigr]^2}{p_k(a)\bigl(1-p_k(a)\bigr)},
  \label{eq:fisher_info}
\end{equation}
and $p_k'(a) = dp_k/da$.
A direct computation gives
\begin{equation}
  I_k(a) = \frac{(2k+1)^2}{a(1-a)},
  \qquad p_k(a)\notin\{0,1\},
  \label{eq:fisher_explicit}
\end{equation}
which simplifies at $a=\tfrac{1}{2}$ to
$I_k\!\bigl(\tfrac{1}{2}\bigr) = 4(2k+1)^2$.
The expression is undefined whenever $p_k(a)\in\{0,1\}$
(i.e.\ when $(2k+1)\arcsin(\sqrt{a})\in\tfrac{\pi}{2}\mathbb{Z}$);
in practice the schedule $\mathcal{K}$ is chosen to avoid these
degenerate amplitudes.
The following proposition makes this criterion explicit.

\begin{proposition}[Admissible MLAE schedule]
\label{prop:admissible_schedule}
Fix $a\in(0,1)$.
A level $k\in\mathbb{N}_0$ is called \emph{degenerate for~$a$} if
\begin{equation}
  (2k+1)\arcsin(\!\sqrt{a})\;\in\;\tfrac{\pi}{2}\,\mathbb{Z},
  \label{eq:degenerate_condition}
\end{equation}
equivalently, if $p_k(a)\in\{0,1\}$.
A schedule $\mathcal{K}\subset\mathbb{N}_0$ is \emph{admissible for~$a$}
if it contains no degenerate level.
For any admissible schedule, $I_k(a)$ is finite and positive for every
$k\in\mathcal{K}$, so the Fisher information~\eqref{eq:fisher_info}
is well-defined.

For fixed $a$, the set of degenerate levels is
$\bigl\{k\in\mathbb{N}_0 : (2k+1) \in \frac{\pi/2}{\arcsin(\sqrt{a})}\mathbb{Z}\bigr\}$,
which is finite or empty for all but the countable set of amplitudes
$a = \sin^2\!\bigl(\frac{j\pi}{2(2k+1)}\bigr)$, $j,k\in\mathbb{N}$.
\end{proposition}

\begin{example}[Degeneracy of $g_0$ at $k=1$]
\label{ex:degeneracy_g0}
For $g_0(x)=\tfrac{1}{4}$, the encoded amplitude is $a=\tfrac{1}{4}$,
so $\arcsin(\!\sqrt{a})=\arcsin(\tfrac{1}{2})=\tfrac{\pi}{6}$.
At level $k=1$: $(2\cdot 1+1)\cdot\tfrac{\pi}{6}=\tfrac{\pi}{2}\in\tfrac{\pi}{2}\mathbb{Z}$,
hence $p_1(\tfrac{1}{4})=\sin^2(\tfrac{\pi}{2})=1$ and $k=1$ is
degenerate.
The admissible schedule $\mathcal{K}=\{0,2\}$ bypasses this level;
$k=0$ and $k=2$ satisfy $p_0(\tfrac{1}{4})=\tfrac{1}{4}$
and $p_2(\tfrac{1}{4})=\sin^2(\tfrac{5\pi}{6})=\tfrac{1}{4}$,
both strictly between $0$ and~$1$.
This is the schedule used for $g_0$ in \cref{sec:experiments}.
\end{example}

\begin{remark}[The two independent error sources]
\label{rem:two_errors}
Inserting the intermediate quantity $Q_n^{(p)}[g]$
and applying the triangle inequality gives
\begin{equation*}
  \bigl|I[g] - \hat{a}\bigr|
  \leq \underbrace{\bigl|I[g] - Q_n^{(p)}[g]\bigr|}_{
         =\,E_n^{(p)}[g],\;\text{discretisation}}
     + \underbrace{\bigl|\hat{a} - Q_n^{(p)}[g]\bigr|}_{
         \text{estimation}}.
\end{equation*}
The two sources are independent: the discretisation error
$E_n^{(p)}[g]$ depends on the regularity of~$g$ and the order
of the quadrature rule, and decreases with~$n$; the estimation
error is controlled by MLAE~\eqref{eq:mlae_mse} and decreases
with~$M$ (with high probability via~\eqref{eq:mlae_hp}).
The gate complexity of~$A_g$ — the cost of a single oracle
call — enters neither bound directly, yet it determines the
practical feasibility of the method.
The central question of this paper is therefore: \emph{for
which functions~$g$ is the projector-expectation
formula~\eqref{eq:oracle_prob} cheap to implement as a
circuit?}
\end{remark}

\section{The Encoding Map}
\label{sec:encoding}

\subsection{Grid and bit-string notation}

Fix $n \geq 1$ and let $\mathbb{B} := \{0,1\}$, so
$\mathbb{B}^n = \{0,1\}^n$ is the set of $n$-bit strings.
For $b = (b_0,\ldots,b_{n-1}) \in \mathbb{B}^n$ define the grid index
\begin{equation}
  i(b) = \sum_{k=0}^{n-1} b_k\, 2^k \in \{0,\ldots,2^n-1\},
  \label{eq:index_map}
\end{equation}
and the grid point $x_{i(b)} = i(b)/2^n \in \mathcal{X}_n$.
Write $[n] := \{0,1,\ldots,n-1\}$ for the set of bit-position indices.

\begin{remark}[Two index sets]
Two index sets appear throughout: bit positions $[n]$ (labelling
coordinates of~$b$ and subsets in multilinear expansions) and grid
indices $\{0,\ldots,2^n-1\}$ (labelling grid points).
The bijection $i:\mathbb{B}^n\to\{0,\ldots,2^n-1\}$ in~\eqref{eq:index_map}
converts between them.
\end{remark}

\subsection{The angle map}

The amplitude oracle $A_g$ is built from an operator $G_g$ whose action
on the ancilla qubit at each grid point $x_{i(b)}$ is a rotation by
angle $\Thetag(b)$.
The precise construction is given in \cref{app:qc}; what matters here is
the following.

\begin{definition}[Angle map]
\label{def:angle_map}
For $g:\mathcal{X}_n\to[0,1]$, the \emph{angle map} is the function
\begin{equation}
  \Thetag:\mathbb{B}^n\to[0,\pi],
  \qquad
  \Thetag(b) = 2\arcsin\!\bigl(\sqrt{g(x_{i(b)})}\bigr).
  \label{eq:angle_map}
\end{equation}
\end{definition}

The map $\phi(t) = 2\arcsin(\sqrt{t})$ is a bijection $[0,1]\to[0,\pi]$
with inverse $\phi^{-1}(\theta) = \sin^2(\theta/2)$.
Since the reindexing $b\mapsto x_{i(b)}$ is also a bijection, the angle
map $g\mapsto\Thetag$ is a bijection from $[0,1]^{\mathcal{X}_n}$ to
$[0,\pi]^{\mathbb{B}^n}$.
In particular, \emph{$g$ and $\Thetag$ carry identical information}.
The significance of this observation is the following key fact.

\begin{proposition}[Encoding cost equals angle-map complexity]
\label{prop:encoding_cost}
The encoding operator $G_g \in \mathrm{U}(2N)$ is block-diagonal with
respect to the decomposition $\mathbb{C}^{2N} = \bigoplus_{i=0}^{N-1}
\mathrm{span}\{\mathbf{u}_i\} \otimes \mathbb{C}^2$,
and the block at index~$i$ is the $2\times 2$ rotation matrix
\begin{equation}
  R_Y(\Thetag(b(i)))
  = \begin{pmatrix}
      \sqrt{1-g(x_i)} & -\sqrt{g(x_i)} \\
      \sqrt{g(x_i)}   & \sqrt{1-g(x_i)}
    \end{pmatrix},
  \label{eq:G_block}
\end{equation}
where $b(i)$ is the unique bit-string with $i(b(i))=i$.
Consequently, an efficient implementation of~$G_g$ is equivalent to
an efficient encoding of the map $b\mapsto\Thetag(b)$.
In particular, any upper bound on the circuit complexity of encoding
$\Thetag:\mathbb{B}^n\to[0,\pi]$ translates directly into an upper
bound on the gate complexity of~$A_g$; the algebraic structure of
$\Thetag$ therefore governs the encoding cost.
\end{proposition}

\begin{proof}
By \eqref{eq:Gg_def}, $G_g$ preserves each summand
$\mathrm{span}\{\mathbf{u}_i\}\otimes\mathbb{C}^2$, so it is
block-diagonal with a $2\times 2$ block at each index~$i$.
Writing $\theta_i := \Thetag(b(i)) = 2\arcsin(\sqrt{g(x_i)})$, the
half-angle identities give
$\cos(\theta_i/2) = \sqrt{1-g(x_i)}$ and
$\sin(\theta_i/2) = \sqrt{g(x_i)}$, so that block is exactly
$R_Y(\theta_i)$, which is~\eqref{eq:G_block}.
Since $A_g = G_g(H_n\otimes I_2)$ by~\eqref{eq:Ag_def} and
$H_n\otimes I_2$ costs $n$ gates independently of~$g$, every
factorisation of $b\mapsto\Thetag(b)$ into controlled rotations yields
a factorisation of $G_g$, and hence of $A_g$, with the same number of
non-trivial factors.
\end{proof}

\Cref{prop:encoding_cost} shifts the problem: instead of asking "which
functions admit low-complexity amplitude encodings?", we ask
the purely combinatorial question "which functions $\Thetag:\mathbb{B}^n\to[0,\pi]$
have low complexity as Boolean functions?".
The answer is given by the multilinear degree theory developed in the
next section.

\subsection{Multilinear structure of the angle map}
\label{subsec:multilinear}

Any function $f:\mathbb{B}^n\to\R$ has a unique representation as a
\emph{multilinear polynomial}:
\begin{equation}
  f(b) = \sum_{S\subseteq[n]} \hat{f}_S \prod_{j\in S} b_j
       = \sum_{S\subseteq[n]} \hat{f}_S\,\chimon{S}(b),
  \label{eq:fourier_expansion}
\end{equation}
where $\chimon{S}(b) := \prod_{j\in S} b_j \in \{0,1\}$ is the
\emph{monomial indicator} of $S\subseteq[n]$, and the coefficients
$\hat{f}_S \in \R$ are the \emph{multilinear coefficients} of~$f$.
The \emph{multilinear degree} is
$\deg(f) = \max\{|S| : \hat{f}_S \neq 0\}$.

The coefficients are recovered by Möbius inversion on the Boolean lattice.

\begin{proposition}[Möbius inversion]
\label{prop:mobius}
For every $f:\mathbb{B}^n\to\R$ and $S\subseteq[n]$,
\begin{equation}
  \hat{f}_S
  = \sum_{T\subseteq S}(-1)^{|S\setminus T|}\,f(\indvec{T})
  = \sum_{b\leq S}(-1)^{|S|-|\mathrm{supp}(b)|}\,f(b),
  \label{eq:mobius}
\end{equation}
where $\indvec{T}\in\mathbb{B}^n$ has $(\indvec{T})_j = \Iver{j\in T}$,
$\mathrm{supp}(b) := \{j\in[n] : b_j = 1\}$ is the support of~$b$,
and $b\leq S$ means $\mathrm{supp}(b)\subseteq S$.
\end{proposition}

\begin{proof}
Substituting \eqref{eq:fourier_expansion} into the right side of
\eqref{eq:mobius} and using $\chimon{R}(\indvec{T}) = \Iver{R\subseteq T}$,
the inner sum evaluates to $\sum_{V\subseteq S\setminus R}(-1)^{|V|}
= (1-1)^{|S\setminus R|} = \Iver{S=R}$, which collapses the double
sum to $\hat{f}_S$.
\end{proof}

\begin{corollary}[Degree characterisation]
\label{cor:degree}
For $f:\mathbb{B}^n\to\R$, $\deg(f)\leq d$ if and only if
\begin{equation}
  \sum_{b\leq S}(-1)^{|S|-|\mathrm{supp}(b)|}\,f(b) = 0
  \qquad\text{for all }S\subseteq[n]\text{ with }|S|>d.
  \label{eq:degree_condition}
\end{equation}
\end{corollary}

\begin{remark}[Computational cost]
\label{rem:wht}
The full coefficient table $\{\hat{f}_S\}_{S\subseteq[n]}$ can be
computed from $\{f(b)\}_{b\in\mathbb{B}^n}$ in $\mathcal{O}(n\,2^n)$
time via the Walsh--Hadamard transform (WHT), a fast algorithm for
Möbius inversion on $2^{[n]}$; see \cref{app:walsh} and
\cite{O'Donnell2014}.
For moderate $n$ (say $n\leq 20$), this is a practical classical
preprocessing step.
\end{remark}

The degree of $\Thetag$ as a multilinear polynomial over $\mathbb{B}^n$
is therefore a classically checkable invariant that, by
\cref{prop:encoding_cost}, controls the gate complexity of the quantum
encoding.
The next section builds the theory around this observation.

\section{Angle-Structure Hierarchy and Circuit Complexity}
\label{sec:hierarchy}

\subsection{Definition and basic structure}

\begin{definition}[Angle-structure class]
\label{def:Gnd}
For $n\geq 1$ and $0\leq d\leq n$, the \emph{angle-structure class of
degree~$d$} is
\begin{equation}
  \Gclass{n}{d} :=
  \bigl\{g:\mathcal{X}_n\to[0,1] : \deg(\Thetag)\leq d\bigr\},
  \label{eq:Gnd_def}
\end{equation}
where $\deg(\Thetag)$ is the multilinear degree of
$\Thetag:\mathbb{B}^n\to[0,\pi]$ defined in \cref{def:angle_map}.
\end{definition}

The filtration
$\Gclass{n}{0}\subseteq\Gclass{n}{1}\subseteq\cdots\subseteq\Gclass{n}{n}$
is strict: $\Gclass{n}{0}$ consists of constant functions; $\Gclass{n}{n}$
contains every $g:\mathcal{X}_n\to[0,1]$ (since any function on
$\mathbb{B}^n$ has a multilinear representation of degree $\leq n$);
and for each $1\leq d\leq n$ the monomial function
$\Thetag(b) = \pi\,\chimon{S}(b)$ for any $|S|=d$ witnesses
$\Gclass{n}{d}\setminus\Gclass{n}{d-1}\neq\emptyset$.

Membership in $\Gclass{n}{d}$ is equivalent (by \cref{cor:degree}) to
the finite system of linear constraints
\begin{equation}
  \sum_{b\leq S}(-1)^{|S|-|\mathrm{supp}(b)|}\,\Thetag(b) = 0
  \qquad\text{for all }S\subseteq[n]\text{ with }|S|>d.
  \label{eq:membership_conditions}
\end{equation}
These conditions involve only the $2^n$ values of the angle table
$\{\Thetag(b)\}_{b\in\mathbb{B}^n}$, which are computable from the
function values $\{g(x_i)\}$ via \eqref{eq:angle_map}.

\subsection{Worked examples}

\begin{example}[$g(x)=\sin^2(\tfrac{\pi x}{2})\in\Gclass{2}{1}$]
\label{ex:sin2}
Let $n=2$, grid $\mathcal{X}_2=\{0,\tfrac{1}{4},\tfrac{1}{2},\tfrac{3}{4}\}$,
index map $i(b_0,b_1)=b_0+2b_1$.
For $g(x)=\sin^2(\tfrac{\pi x}{2})$, using
$\Thetag(b)=2\arcsin(\sin(\tfrac{\pi x_{i(b)}}{2}))=\pi x_{i(b)}$:
\begin{equation*}
  \Thetag(0,0)=0,\quad
  \Thetag(1,0)=\tfrac{\pi}{4},\quad
  \Thetag(0,1)=\tfrac{\pi}{2},\quad
  \Thetag(1,1)=\tfrac{3\pi}{4}.
\end{equation*}
Applying Möbius inversion \eqref{eq:mobius}:
\begin{equation*}
  \wh{\Theta}{\emptyset}=0,\quad
  \wh{\Theta}{\{0\}}=\tfrac{\pi}{4},\quad
  \wh{\Theta}{\{1\}}=\tfrac{\pi}{2},\quad
  \wh{\Theta}{\{0,1\}}=\tfrac{3\pi}{4}-\tfrac{\pi}{2}-\tfrac{\pi}{4}+0=0.
\end{equation*}
Since $\wh{\Theta}{\{0,1\}}=0$, the angle map is exactly affine:
$\Thetag(b_0,b_1)=\tfrac{\pi}{4}b_0+\tfrac{\pi}{2}b_1$, confirming
$g\in\Gclass{2}{1}$.
\end{example}

\begin{example}[$g(x)=\sin^2(\pi x)\in\Gclass{2}{2}\setminus\Gclass{2}{1}$]
\label{ex:sin2_deg2}
Same grid. Function values:
$g(0)=0$, $g(\tfrac{1}{4})=\tfrac{1}{2}$,
$g(\tfrac{1}{2})=1$, $g(\tfrac{3}{4})=\tfrac{1}{2}$.
Angle values: $\Thetag(0,0)=0$, $\Thetag(1,0)=\tfrac{\pi}{2}$,
$\Thetag(0,1)=\pi$, $\Thetag(1,1)=\tfrac{\pi}{2}$.
The degree-2 Möbius coefficient is
$\wh{\Theta}{\{0,1\}}=\tfrac{\pi}{2}-\pi-\tfrac{\pi}{2}+0=-\pi\neq 0$,
so $\deg(\Thetag)=2$ and
$\Thetag(b_0,b_1)=\tfrac{\pi}{2}b_0+\pi b_1-\pi b_0 b_1$.
This witnesses $g\in\Gclass{2}{2}\setminus\Gclass{2}{1}$.
\end{example}

\subsection{Characterisation of \texorpdfstring{$\Gclass{n}{1}$}{G\{n\}\{1\}}}

The affine class $\Gclass{n}{1}$ is the most relevant for practice,
as it yields the shallowest circuits.

\begin{theorem}[Characterisation of $\Gclass{n}{1}$]
\label{thm:affine_char}
A function $g:\mathcal{X}_n\to[0,1]$ belongs to $\Gclass{n}{1}$ if and
only if the angle values $\theta_b:=\Thetag(b)$, $b\in\mathbb{B}^n$,
satisfy the \emph{affine consistency conditions}:
\begin{equation}
  \sum_{b\in\mathbb{B}^n}(-1)^{|S\cap\mathrm{supp}(b)|}\,\theta_b = 0
  \qquad\text{for all }S\subseteq[n]\text{ with }|S|\geq 2,
  \label{eq:affine_conditions}
\end{equation}
equivalently, $\wh{\Theta}{S}=0$ for all $|S|\geq 2$.
\end{theorem}

\begin{proof}
This is the instance $d=1$ of \cref{cor:degree}, with the vanishing
conditions $\wh{\Theta}{S}=0$ for $|S|\geq 2$ rewritten
in inclusion-exclusion form via \eqref{eq:mobius}.
\end{proof}

The conditions \eqref{eq:affine_conditions} can be verified as follows:
first compute the full Walsh--Hadamard coefficient table
$\{\wh{\Theta}{S}\}_{S\subseteq[n]}$ in $\mathcal{O}(n\,2^n)$ time
(cf.\ \cref{rem:wht}), then check that $\wh{\Theta}{S}=0$ for all
$|S|\geq 2$, which requires a single $\mathcal{O}(2^n)$ pass over the
table.
The total preprocessing cost is therefore $\mathcal{O}(n\,2^n)$,
feasible for $n\leq 20$ before any quantum computation.

\subsection{The encoding circuit theorem}

We now state and prove the main structural result, which identifies
$\deg(\Thetag)$ as a classically computable invariant that controls
the gate count of the canonical monomial-factorisation circuit for
the encoding operator.

\begin{theorem}[Encoding circuit via monomial factorisation]
\label{thm:circuit_depth}
Let $g\in\Gclass{n}{d}$, so $\Thetag$ has the multilinear expansion
\begin{equation}
  \Thetag(b)
  = \sum_{\substack{S\subseteq[n]\\ |S|\leq d}}
    \wh{\Theta}{S}\,\chimon{S}(b),
  \qquad
  \wh{\Theta}{S}
  = \sum_{T\subseteq S}(-1)^{|S\setminus T|}\,\Thetag(\indvec{T}).
  \label{eq:theta_expansion}
\end{equation}
Then the encoding operator $G_g$ factors as
\begin{equation}
  G_g
  = \prod_{\substack{S\subseteq[n]\\ |S|\leq d}}
    C^S\Ry\!\bigl(\wh{\Theta}{S}\bigr),
  \label{eq:G_factorization}
\end{equation}
where $C^S\Ry(\alpha)$ denotes the $\Ry(\alpha)$ rotation on the ancilla
controlled on the condition that all index qubits in~$S$ are in state~$1$.
The encoding operator can be implemented with exactly
\begin{equation}
  \sum_{k=0}^{d}\binom{n}{k}
  \label{eq:gate_count}
\end{equation}
controlled-$\Ry$ gates.
\end{theorem}

\begin{proof}
The proof has three steps: we compute the action of the proposed
product circuit on each index-register fibre separately, assemble the
global operator by block-diagonal extension, and identify it with
$G_g$.

\medskip
\noindent\textbf{Step 1: Fibre-wise action.}
Since $g\in\Gclass{n}{d}$, \cref{prop:mobius,cor:degree} give the
unique multilinear expansion \eqref{eq:theta_expansion} with only
terms of degree $\leq d$.
The $(n+1)$-qubit space decomposes as
\begin{equation}
  \mathbb{C}^{2^{n+1}}
  = \mathbb{C}^N \otimes \mathbb{C}^2
  = \bigoplus_{b\in\mathbb{B}^n} V_b,
  \qquad V_b := \mathrm{span}\{\mathbf{u}_b\} \otimes \mathbb{C}^2,
  \label{eq:fibre_decomp}
\end{equation}
where $N=2^n$ and each fibre $V_b$ is the two-dimensional subspace
spanned by $\{\mathbf{u}_b\otimes\mathbf{u}_0^{(2)},\,
\mathbf{u}_b\otimes\mathbf{u}_1^{(2)}\}$.

Fix $b\in\mathbb{B}^n$.
By definition of $C^S\Ry(\alpha)$, its restriction to $V_b$ acts on
the ancilla factor as
\[
  C^S\Ry(\alpha)\big|_{V_b}
  =
  \begin{cases}
    I_{2^n}\otimes R_Y(\alpha) & \text{if } \chimon{S}(b)=1,\\
    I_{2^n}\otimes I_2          & \text{if } \chimon{S}(b)=0.
  \end{cases}
\]
Thus, restricted to $V_b$, every factor of the product
\eqref{eq:G_factorization} acts on the ancilla $\mathbb{C}^2$ as
a rotation $R_Y(\wh{\Theta}{S}\,\chimon{S}(b))$ about the $Y$-axis.
Since $\{R_Y(\alpha):\alpha\in\R\}$ is an abelian subgroup of
$\mathrm{U}(2)$ — satisfying $R_Y(\alpha)R_Y(\beta)=R_Y(\alpha+\beta)$
for all $\alpha,\beta$ — the restriction of the full product to $V_b$
is, in any order,
\begin{equation}
  \prod_{\substack{S\subseteq[n]\\ |S|\leq d}}
  C^S\Ry\!\bigl(\wh{\Theta}{S}\bigr)\Bigg|_{V_b}
  = I_{2^n}\otimes
    R_Y\!\Biggl(
      \sum_{\substack{S\subseteq[n]\\ |S|\leq d}}
      \wh{\Theta}{S}\,\chimon{S}(b)
    \Biggr)
  = I_{2^n}\otimes R_Y\!\bigl(\Thetag(b)\bigr),
  \label{eq:fibre_product}
\end{equation}
where the last equality uses \eqref{eq:theta_expansion}.
In particular, on the basis vectors of $V_b$:
\begin{equation}
  \prod_{\substack{S\subseteq[n]\\ |S|\leq d}}
  C^S\Ry\!\bigl(\wh{\Theta}{S}\bigr)\,
  \bigl(\mathbf{u}_b\otimes\mathbf{u}_0^{(2)}\bigr)
  = \mathbf{u}_b\otimes R_Y\!\bigl(\Thetag(b)\bigr)\,
    \mathbf{u}_0^{(2)}.
  \label{eq:fibre_action}
\end{equation}

\medskip
\noindent\textbf{Step 2: Global operator by block-diagonal extension.}
Every gate $C^S\Ry(\alpha)$ leaves the index register unchanged and
acts only on the ancilla; it therefore maps each fibre $V_b$ into
itself.
Consequently the product operator \eqref{eq:G_factorization} is
block-diagonal with respect to the decomposition
\eqref{eq:fibre_decomp}, and its action on all of
$\mathbb{C}^{2^{n+1}}$ is completely determined by the fibre-wise
actions~\eqref{eq:fibre_product}.

\medskip
\noindent\textbf{Step 3: Identification with $G_g$.}
By \cref{prop:encoding_cost,eq:G_block} and the definition of $G_g$
in \cref{app:qc}, the operator $G_g$ is also block-diagonal with
respect to \eqref{eq:fibre_decomp}, with block at index~$b$ equal to
$R_Y(\Thetag(b))$ acting on the ancilla factor.
Comparing with \eqref{eq:fibre_action}, the product operator and $G_g$
agree on every fibre basis vector, hence on a basis of
$\mathbb{C}^{2^{n+1}}$, and therefore as operators on
$\mathbb{C}^{2^{n+1}}$:
\[
  \prod_{\substack{S\subseteq[n]\\ |S|\leq d}}
  C^S\Ry\!\bigl(\wh{\Theta}{S}\bigr)
  \;=\; G_g.
\]

\medskip
\noindent\textbf{Gate count and circuit depth.}
The product \eqref{eq:G_factorization} contains one factor
$C^S\Ry(\wh{\Theta}{S})$ for each subset $S\subseteq[n]$ with
$|S|\leq d$; the number of such subsets is $\sum_{k=0}^{d}\binom{n}{k}$.

The gates commute in the following precise sense: on each fibre $V_b$,
every factor $C^S\Ry(\wh{\Theta}{S})$ acts on the ancilla as a
rotation $R_Y(\wh{\Theta}{S}\,\chimon{S}(b))$, and since
$\{R_Y(\alpha):\alpha\in\R\}$ is abelian the fibre-wise order is
irrelevant (as used in Step~1 above).
Globally, however, all $\binom{n}{\leq d}$ gates share the same
ancilla qubit as their target, so they \emph{cannot} be executed in
parallel: each gate occupies a separate circuit layer.
The circuit depth therefore equals the gate count
$\binom{n}{\leq d}$, not a smaller parallelised value.
The affine case $d=1$ is worked out explicitly in
\cref{app:circuit_proof}: the $n+1$ gates commute (fibre-wise) but
must be executed sequentially, giving depth $n+1=\mathcal{O}(n)$.
\end{proof}

\begin{remark}[Upper bound, and tightness at $d=n$]
\label{rem:interpolation}
\cref{thm:circuit_depth} gives an \emph{upper bound} on the gate
complexity of the encoding operator~$G_g$: the monomial-factorisation
circuit achieves $\binom{n}{\leq d}$ gates, but a more efficient
implementation using ancilla qubits, gate cancellations, or structured
decompositions might in principle do better for $d < n$.
The gate count interpolates across the hierarchy:
\begin{itemize}
  \item $d=0$ (constant): $1$ gate (one unconditional rotation).
  \item $d=1$ (affine): $n+1$ gates — linear in the number of qubits.
  \item $d=2$ (quadratic): $1 + n + \binom{n}{2} = \mathcal{O}(n^2)$ gates.
  \item $d=n$ (generic): $2^n$ gates.
\end{itemize}
At $d=n$ the upper bound $\binom{n}{\leq n}=2^n$ coincides with the
general lower bound for arbitrary $n$-qubit state preparation
\cite{ShendeBullockMarkov2006}, so \cref{thm:circuit_depth} is
\emph{tight at $d=n$}: the monomial factorisation is optimal for
generic functions.
For $1\leq d < n$, tightness is an open question; see
Open Problem~(i) in \cref{sec:conclusions}.
For intermediate degrees, $\binom{n}{\leq d}=\mathcal{O}(n^d/d!)$,
giving a gate count that grows polynomially in~$n$ for any fixed~$d$.
The degree~$d$ thus plays the role of a \emph{circuit-complexity
exponent} for the canonical encoding circuit.
\end{remark}

\begin{corollary}[Multivariate extension]
\label{cor:multivariate}
Let $D\geq 1$ and let $n_1,\ldots,n_D\geq 1$ with $n=n_1+\cdots+n_D$.
For $g:\mathcal{X}_{n_1}\times\cdots\times\mathcal{X}_{n_D}\to[0,1]$
discretised on the tensor-product grid, define the angle map
$\Thetag:\mathbb{B}^n\to[0,\pi]$ by treating the $n$~index qubits as
a single bit-string $b=(b^{(1)},\ldots,b^{(D)})\in\mathbb{B}^n$.
If $\deg(\Thetag)\leq d$, then the encoding operator $G_g$ can be
implemented with exactly $\sum_{k=0}^{d}\binom{n}{k}$ controlled-$R_Y$
gates, regardless of how the $n$ qubits are partitioned across
dimensions.
In particular, for $d=1$ the gate count is $n+1=n_1+\cdots+n_D+1$,
growing linearly in the total number of qubits.
\end{corollary}

\begin{proof}
\cref{thm:circuit_depth} applies verbatim with $n=n_1+\cdots+n_D$:
the proof depends only on $\deg(\Thetag)\leq d$ and the total qubit
count~$n$, not on the partition of bits across dimensions.
\end{proof}

\subsection{Closure under products and sums}

The hierarchy is compatible with the multiplication and addition
circuits of \cite[Sec.~III]{CarreraVazquezWoerner2020}.
For the product circuit, which estimates $\mathbb{E}[f\cdot g]$ by
applying $F$ and $G$ to two independent ancilla qubits, the two
encoding operators act on separate ancilla factors so no cross-degree
monomials are introduced: for $f\in\Gclass{n}{d_1}$ and
$g\in\Gclass{n}{d_2}$, the gate counts $\binom{n}{\leq d_1}$ and
$\binom{n}{\leq d_2}$ are additive and neither degree is increased.
For the addition circuit, which estimates $\mathbb{E}[g+h]$ via a
shared ancilla qubit and an auxiliary bit $\beta\in\{0,1\}$, the
effective angle map $(b,\beta)\mapsto(1-\beta)\,\Thetag(b)+\beta\,\Theta_h(b)$
is multilinear of degree $\max(d_1,d_2)$, so the construction does not
raise the degree of either factor.
In both cases, neither operation increases the multilinear degree of
the angle map.

\section{Depth-versus-Accuracy Trade-off}
\label{sec:pareto}

The total integration error has two independent sources
(cf.\ \cref{rem:two_errors}): discretisation error~$E_n^{(p)}[g]$
from using a finite grid, and estimation error $\mathcal{O}(1/M)$ from
finite sampling.
The gate count per oracle call is $\binom{n}{\leq d}$ by
\cref{thm:circuit_depth}.
The interplay between these three quantities — degree~$d$, qubit
count~$n$, and shot budget~$M$ — defines a
\emph{trade-off} in the (depth, error) space.

\subsection{Discretisation-error bounds}

Following \cite{CarreraVazquezWoerner2020},
the three standard quadrature rules satisfy
\begin{align}
  E_n^{\mathrm{left}}[g] &\leq \frac{1}{2}\,
    \frac{\|g'\|_\infty}{2^n},
    \qquad g\in C^1[0,1],
    \label{eq:error_left}\\
  E_n^{\mathrm{mid}}[g] &\leq \frac{1}{24}\,
    \frac{\|g''\|_\infty}{2^{2n}},
    \qquad g\in C^2[0,1],
    \label{eq:error_mid}\\
  E_n^{\mathrm{Simpson}}[g] &\leq \frac{1}{2880}\,
    \frac{\|g^{(4)}\|_\infty}{2^{4n}},
    \qquad g\in C^4[0,1],
    \label{eq:error_simp}
\end{align}
corresponding to quadrature orders $p=1,2,4$ respectively.
In general, for a quadrature rule of order~$p\geq 1$ applied to
$g\in C^\alpha[0,1]$ with $\alpha\geq p$,
\begin{equation}
  E_n^{(p)}[g] \leq C_p\,2^{-pn}\,\|g^{(p)}\|_\infty
  \label{eq:disc_general}
\end{equation}
for some constant $C_p>0$ depending only on~$p$.
The factor $\|g^{(p)}\|_\infty$ is finite since $g^{(p)}$
is continuous on the compact interval $[0,1]$ (as $\alpha\geq p$).
The bound~\eqref{eq:disc_general} follows by applying the local
quadrature error estimate on each subinterval $I_i=[i/N,(i+1)/N)$
of length $h=2^{-n}$, which contributes $\mathcal{O}(h^{p+1})$,
and summing over the $N=2^n$ subintervals.
The regularity parameter~$\alpha$ and the quadrature order~$p$ are
thus independent: $\alpha$ is a property of the integrand~$g$,
while $p$ is a property of the chosen numerical scheme;
the bound~\eqref{eq:disc_general} holds whenever $\alpha\geq p$.

\subsection{The trade-off}

\begin{theorem}[Depth-accuracy trade-off]
\label{thm:pareto}
Let $g:[0,1]\to[0,1]$ be such that $g|_{\mathcal{X}_n}\in\Gclass{n}{d}$
and $g\in C^\alpha[0,1]$ with $\alpha\geq p\geq 1$, where $p$ is the
order of the chosen quadrature rule.
Assume $E_n^{(p)}[g]\leq C_p\,2^{-pn}\,\|g^{(p)}\|_\infty$
as in~\eqref{eq:disc_general}.
Fix any $\delta\in(0,1)$.
To achieve total integration error at most~$\varepsilon$ with
probability at least $1-\delta$, the following resources suffice:
\begin{equation}
  n \geq \frac{1}{p}\log_2\!\Bigl(\frac{2C_p\|g^{(p)}\|_\infty}{\varepsilon}\Bigr),
  \qquad
  M \geq \frac{2C_{\mathrm{est}}}{\varepsilon\sqrt{\delta}},
  \qquad
  \text{gate count per oracle call}= \binom{n}{\leq d}.
  \label{eq:pareto_resources}
\end{equation}
Substituting the minimal $n^*=\lceil\frac{1}{p}\log_2(2C_p\|g^{(p)}\|_\infty/\varepsilon)\rceil$
and taking $\delta$ a fixed constant (e.g.\ $\delta=1/4$),
the \emph{total} gate count (across all $M$ oracle calls) as a function
of $\varepsilon$ alone is
\begin{equation}
  \mathcal{O}\!\left(
    \Bigl(\frac{\log(1/\varepsilon)}{p}\Bigr)^d
    \cdot \frac{1}{\varepsilon}
  \right),
  \label{eq:total_gates_eps}
\end{equation}
using $\binom{n}{\leq d}=\mathcal{O}(n^d/d!)$ for fixed~$d$.
\end{theorem}

\begin{proof}
We verify the three resource bounds in~\eqref{eq:pareto_resources} and
then derive~\eqref{eq:total_gates_eps}.

The condition $E_n^{(p)}[g]\leq\varepsilon/2$ is equivalent to
$2^{pn}\geq 2C_p\|g^{(p)}\|_\infty/\varepsilon$, which gives
$n\geq\frac{1}{p}\log_2(2C_p\|g^{(p)}\|_\infty/\varepsilon)$.
The minimal integer satisfying this is
\[
n^*=\Bigl\lceil\tfrac{1}{p}\log_2\!\bigl(2C_p\|g^{(p)}\|_\infty/\varepsilon\bigr)\Bigr\rceil
=\mathcal{O}\!\bigl(\tfrac{1}{p}\log(1/\varepsilon)\bigr).
\]
By construction of the oracle $G_g$ from $g|_{\mathcal{X}_{n^*}}$
(cf.\ \cref{subsec:oracle}),
the encoded amplitude equals $a = Q_{n^*}^{(p)}[g]$.
By~\eqref{eq:mlae_hp}, the MLAE estimator with a linear schedule
satisfies
\[
  \mathbb{P}\!\Bigl(|\hat{a} - Q_{n^*}^{(p)}[g]| \leq
  \tfrac{C_{\mathrm{est}}}{M\sqrt{\delta}}\Bigr) \geq 1-\delta.
\]
Requiring $C_{\mathrm{est}}/(M\sqrt{\delta})\leq\varepsilon/2$ gives
$M\geq 2C_{\mathrm{est}}/(\varepsilon\sqrt{\delta}) = \mathcal{O}(1/\varepsilon)$
for any fixed $\delta>0$.
By \cref{thm:circuit_depth}, $G_g$ requires $\binom{n}{\leq d}$
controlled-$\Ry$ gates per oracle call, establishing the third bound
in~\eqref{eq:pareto_resources}.
On the event $\{|\hat{a}-Q_{n^*}^{(p)}[g]|\leq\varepsilon/2\}$,
which holds with probability at least $1-\delta$, the total error
satisfies
\[
  \bigl|I[g]-\hat{a}\bigr|
  \leq E_{n^*}^{(p)}[g] + \bigl|\hat{a}-Q_{n^*}^{(p)}[g]\bigr|
  \leq \frac{\varepsilon}{2}+\frac{\varepsilon}{2} = \varepsilon
\]
by the triangle inequality.

For the total gate count, note that all $M$ oracle calls together use
$\binom{n^*}{\leq d}\cdot M$ controlled-$\Ry$ gates.
The standard binomial bound $\binom{m}{\leq d}\leq(em/d)^d$,
where $e$ is Euler's number, gives
$\binom{n^*}{\leq d}=\mathcal{O}((n^*)^d)
=\mathcal{O}\bigl((\tfrac{1}{p}\log\tfrac{1}{\varepsilon})^d\bigr)$,
and multiplying by $M=\mathcal{O}(1/\varepsilon)$ yields~\eqref{eq:total_gates_eps}.
\end{proof}

\begin{remark}[Reading the trade-off formula]
\label{rem:reading_tradeoff}
Two aspects of \cref{thm:pareto} deserve explicit comment.

\medskip
\noindent\textit{(i) Role of $\delta$.}
The asymptotic formula~\eqref{eq:total_gates_eps} is stated after fixing
$\delta$ as a constant, so $1/\sqrt{\delta}$ is absorbed into the
$\mathcal{O}(\cdot)$ implicit constant.
The full $\delta$-dependent total gate count, read directly from
\eqref{eq:pareto_resources}, is
$\mathcal{O}\bigl((\log(1/\varepsilon)/p)^d/(\varepsilon\sqrt{\delta})\bigr)$.
Thus $\varepsilon$ and $\delta$ enter independently: tightening the
success probability by decreasing $\delta$ costs an extra factor of
$\delta^{-1/2}$ in gate count, independently of the target accuracy.

\medskip
\noindent\textit{(ii) Interpolation across degrees.}
Formula~\eqref{eq:total_gates_eps} interpolates between two extremes:
at $d=1$ the total cost is $\mathcal{O}(\varepsilon^{-1}\log(1/\varepsilon))$,
essentially linear in the shot budget; at $d=n^*$ it becomes
$\mathcal{O}(\varepsilon^{-(1+1/p)})$, the same order as classical
quadrature without the Monte Carlo overhead.
For intermediate degrees $1<d<n^*$ the extra cost over the $d=1$ case
is the polylogarithmic factor $(\log(1/\varepsilon)/p)^{d-1}$, which is
the quantitative signature of the encoding degree in the circuit-resource
sense.
\end{remark}

At the same discretisation level, a classical $p$-th order quadrature
combined with $M$ Monte Carlo repetitions requires total evaluation
count $N\cdot M=\mathcal{O}(\varepsilon^{-(2+1/p)})$.
The quantum gate count~\eqref{eq:total_gates_eps} for $d=1$ is
$\mathcal{O}(\varepsilon^{-1}\log(1/\varepsilon))$, which lies below
this fixed-discretisation Monte Carlo curve for all sufficiently
small~$\varepsilon$ and every $p\geq 1$; the crossover point depends on
the implicit constants $C_p$ and $C_{\mathrm{est}}$ in
\eqref{eq:pareto_resources}.  Example~\ref{ex:pareto_rules} below
makes this comparison explicit for the three standard rules, and the
curves are illustrated in \cref{fig:pareto}.

\begin{example}[Trade-off for the three standard rules]
\label{ex:pareto_rules}
For $g:[0,1]\to[0,1]$ with $g|_{\mathcal{X}_n}\in\Gclass{n}{1}$
and $g\in C^4[0,1]$, the minimal qubit count and total gate count
for each rule are:
\begin{center}
\renewcommand{\arraystretch}{1.6}
\begin{tabular}{lccc}
\toprule
Rule & Order $p$ & Minimal qubits $n^*$ & Total gates \\
\midrule
Left Riemann  & $1$ &
  $\bigl\lceil\log_2(\|g'\|_\infty/\varepsilon)\bigr\rceil$ &
  $\mathcal{O}(\varepsilon^{-1}\log(\|g'\|_\infty/\varepsilon))$ \\[4pt]
Midpoint      & $2$ &
  $\bigl\lceil\tfrac{1}{2}\log_2(\|g''\|_\infty/12\varepsilon)\bigr\rceil$ &
  $\mathcal{O}(\varepsilon^{-1}\log(\|g''\|_\infty/\varepsilon))$ \\[4pt]
Simpson       & $4$ &
  $\bigl\lceil\tfrac{1}{4}\log_2(\|g^{(4)}\|_\infty/720\varepsilon)\bigr\rceil$ &
  $\mathcal{O}(\varepsilon^{-1}\log(\|g^{(4)}\|_\infty/\varepsilon))$ \\
\bottomrule
\end{tabular}
\end{center}
All three rules achieve the same asymptotic gate count
$\mathcal{O}(\varepsilon^{-1}\log(1/\varepsilon))$ for $d=1$, but
higher-order rules require fewer qubits (by factors of $1/2$ and $1/4$),
directly reducing the constant in the gate count through the smaller
value of~$n^*$.
\end{example}

\medskip
\noindent\cref{thm:pareto} requires $g\in C^\alpha[0,1]$ with $\alpha\geq p\geq 1$.
The intermediate Sobolev regime $s\in(1/2,1)$, where $g\in W^{s,2}(0,1)$
but $g\notin C^1[0,1]$ in general, is covered by a straightforward
extension.
Throughout we take $s'\in(1/2,s)$, so that the embedding
$W^{s',2}(0,1)\hookrightarrow C[0,1]$ holds and point evaluations of~$g$
are well defined; this is the standing requirement for
any quadrature rule based on function values, quantum or classical
\cite[Chap.~4]{NovakWozniakowski2008}.
Under this hypothesis the midpoint rule on the $2^n$-point grid
satisfies
\begin{equation}
  E_n^{\mathrm{mid}}(g) \leq C_{s'}\,2^{-s'n}\,\|g\|_{W^{s',2}},
  \label{eq:disc_sobolev_mid}
\end{equation}
which is the order-optimal deterministic rate $N^{-s'}$ for the class
$B_{W^{s',2}}$ \cite{NovakWozniakowski2008}.
Hence \cref{thm:pareto} applies with $p$ replaced by $s'$, yielding
$n^*=\mathcal{O}(\log(1/\varepsilon))$ as before and total gate count
$\mathcal{O}(\varepsilon^{-1}(\log 1/\varepsilon)^d)$; only the
constant in $n^*$ changes.
The classical Monte Carlo cost in this regime is
$\mathcal{O}(\varepsilon^{-2})$, so the quadratic gain of QAE in the
estimation stage persists once the encoding is affine.

\begin{remark}[Scope of the comparison]
\label{rem:scope}
The comparison above is between the \emph{gate count} of the QAE
pipeline and the \emph{sample count} of classical Monte Carlo, at a
fixed discretisation level~$n$.
It is not a statement in the information-based complexity model:
the $2^n$ classical evaluations of~$g$ consumed when the angle table
is precomputed are of exactly the same nature as the evaluations
charged to a classical quadrature rule, and the bias term
$E_n^{\mathrm{mid}}(g)$ is subject to the same deterministic lower
bounds over Sobolev balls as any $N$-point rule
\cite{Novak2001,NovakWozniakowski2008}.
What \cref{thm:pareto} asserts is that, for $g\in\Gclass{n}{d}$ with
small~$d$, the \emph{quantum} part of the pipeline does not add a
dominant overhead --- not that the discretisation requirement can be
circumvented.
For the optimal rates of integration over Sobolev balls in the
deterministic, randomised and quantum query models we refer to
\cite{Novak2001,Heinrich2002,Heinrich2003}.
\end{remark}

\begin{figure}[t]
\centering
\includegraphics[width=0.78\textwidth]{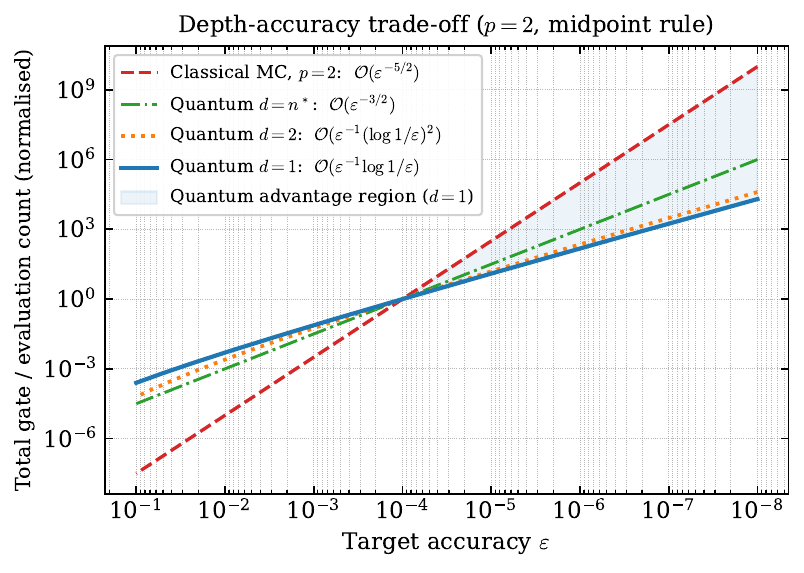}
\caption{Depth-accuracy trade-off: total gate count (quantum)
  or evaluation count (classical) normalised to~$1$ at $\varepsilon=10^{-4}$,
  plotted against target accuracy~$\varepsilon$ for the midpoint rule
  ($p=2$).
  Quantum curves for $d=1$ (solid blue), $d=2$ (dotted orange), and
  $d=n^*$ (dash-dot green) are given by \eqref{eq:total_gates_eps};
  the classical Monte Carlo curve (dashed red) scales as
  $\mathcal{O}(\varepsilon^{-5/2})$.
  The shaded region marks where the quantum $d=1$ curve lies below the
  fixed-discretisation Monte Carlo curve.
  For $d=1$ the quantum gate count grows as
  $\mathcal{O}(\varepsilon^{-1}\log(1/\varepsilon))$,
  below the corresponding Monte Carlo sample count in this restricted
  comparison for any $p \geq 1$.}
\label{fig:pareto}
\end{figure}

\subsection{Encoding degree versus Sobolev regularity}
\label{subsec:decoupling}

The gate count of \cref{thm:circuit_depth} is governed by the
multilinear degree of $\Thetag:\mathbb{B}^n\to[0,\pi]$, which is a
property of the \emph{discrete} angle table, not of the
\emph{continuous} function~$g$.
The two are in fact independent: one can construct functions~$g$ of
arbitrarily low Sobolev regularity that nevertheless have an angle map
of degree~$1$, so that the encoding circuit stays at depth
$\mathcal{O}(n)$.
This subsection makes that decoupling precise.
It is a structural statement about the hierarchy
$\{\Gclass{n}{d}\}$; no complexity comparison is drawn from it
(cf.\ \cref{rem:scope}).

Recall that the fractional Sobolev space $W^{s,2}(0,1)$, $s\in(0,1)$,
consists of functions $u\in L^2(0,1)$ with finite Aronszajn--Slobodeckij
seminorm
\begin{equation}
  |u|_{W^{s,2}}^2
  = \int_0^1\!\int_0^1\frac{|u(x)-u(y)|^2}{|x-y|^{1+2s}}\,dx\,dy < \infty.
  \label{eq:sobolev_seminorm}
\end{equation}

\begin{lemma}[Affine-angle functions with prescribed Sobolev regularity]
\label{lem:gs_construction}
For $n\geq 2$ and $s\in(0,1/2)$, there exists
$g_s:[0,1]\to[0,1]$ satisfying:
\begin{enumerate}[label=\textup{(\roman*)}]
\item $g_s|_{\mathcal{X}_n}\in\Gclass{n}{1}$: the angle table
  $\{\Theta_{g_s|_{\mathcal{X}_n}}(b) : b\in\mathbb{B}^n\}$ satisfies all affine consistency conditions
  \eqref{eq:affine_conditions}.
\item $g_s\in W^{s',2}(0,1)$ for all $s'<s$, but
  $g_s\notin W^{s,2}(0,1)$.
\end{enumerate}
The function $g_s$ is continuous on the interior of each grid cell and
may have jump discontinuities at grid interfaces; the construction
below ensures that at least one such jump occurs.
It is this jump structure that confines the statement to $s<1/2$: the
interface contribution of a bounded jump is finite precisely when
$s'<1/2$, and it diverges otherwise.
\end{lemma}

\begin{proof}
\textit{Construction of the angle table and grid values.}
Choose $\gamma\in\R$ and $(c_0,\ldots,c_{n-1})\in\R^n$ with all
$c_k\neq 0$ so that
$0\leq\gamma+\sum_{k=0}^{n-1}c_k b_k\leq\pi$ for every
$b\in\mathbb{B}^n$, and define
\begin{equation}
  \Theta_{g_s|_{\mathcal{X}_n}}(b) = \gamma+\sum_{k=0}^{n-1}c_k b_k,
  \qquad b\in\mathbb{B}^n.
  \label{eq:affine_angle_pf}
\end{equation}
This is a multilinear polynomial of degree~$1$ by construction, so the
affine consistency conditions~\eqref{eq:affine_conditions} hold and
(i) is satisfied.
Set $v_i := \sin^2(\Theta_{g_s|_{\mathcal{X}_n}}(b(i))/2) \in [0,1]$ for each
$i \in \{0,\ldots,N-1\}$; these are the target values of $g_s$ at the
grid points.

\textit{Weierstrass-type interpolation.}
Partition $[0,1]$ into $N = 2^n$ half-open intervals
$I_i = [i/N, (i+1)/N)$.
We take $\phi:\R\to[0,1]$ to be the $1$-periodic function defined by
$\phi(t) = \sin^2(\pi t)$ on $[0,1]$ and extended by periodicity.
This satisfies $\phi(0)=0$, $\phi(1/2)=1$, and
$\phi(k)=0$ for every $k\in\mathbb{Z}$.
Set $\phi_m(t) := 2^{-ms}\phi(2^m t)$ and define the Weierstrass-type
series
\begin{equation}
  w_s(t) = \frac{1}{Z_s}\sum_{m=0}^{\infty}\phi_m(t),
  \qquad t\in[0,1],
  \label{eq:weierstrass_pf}
\end{equation}
where
\[
  Z_s := \sum_{m=0}^{\infty} 2^{-ms} = \frac{1}{1-2^{-s}} > 0
  \qquad (s > 0);
\]
the series in~\eqref{eq:weierstrass_pf} converges uniformly since
$|\phi_m(t)|\leq 2^{-ms}$ and $\sum_m 2^{-ms}<\infty$
(Weierstrass $M$-test \cite{Zygmund2002}).
Since $\phi(k)=0$ for all $k\in\mathbb{Z}$, we have
$\phi_m(0) = 2^{-ms}\phi(0) = 0$ for every $m\geq 0$,
hence $w_s(0) = 0$.
Since $\phi(1/2)=1$ and $\phi_m(1/2) = 2^{-ms}\sin^2(\pi\cdot 2^{m-1})=0$
for all $m\geq 1$ (as $2^{m-1}\in\mathbb{Z}$), the series satisfies
$w_s(1/2) = Z_s^{-1}\phi_0(1/2)=Z_s^{-1}$; in particular,
$w_s$ is non-constant (it takes the value $0$ at $t=0$ and
a strictly positive value at $t=1/2$), which is all that is needed below.%
\footnote{The specific value $w_s(1/2)=Z_s^{-1}$ follows from
$\phi_m(1/2)=0$ for $m\geq 1$ and $\phi_0(1/2)=1$, giving
$w_s(1/2) = Z_s^{-1}$.
However, only $w_s(0)=0$ and non-constancy are used in the proof.}
On each interval $I_i$ define
\begin{equation}
  g_s(x) = v_i + (v_{i+1}-v_i)\cdot w_s(2^n x - i),
  \qquad x\in I_i,
  \label{eq:gs_def}
\end{equation}
with the convention $v_N := v_0$ for the last interval.
Set also $g_s(1):=v_N$.
Since $w_s(0) = 0$, we have $g_s(i/N) = v_i$ for every
$i \in \{0,\ldots,N-1\}$, so the grid values of $g_s$ match the
prescribed $v_i$ and the angle table is exactly \eqref{eq:affine_angle_pf},
confirming (i).

\textit{Sobolev regularity.}
We prove directly that $w_s \in W^{s',2}(0,1)$ for every $s'<s$
and $w_s\notin W^{s,2}(0,1)$.
Since $W^{s',2}(0,1) = H^{s'}(0,1)$ for $s'\in(0,1)$, the squared
norm $\|f\|^2_{W^{s',2}}$ is equivalent to
\begin{equation}
  \sum_{k\in\mathbb{Z}} (1+|k|^2)^{s'}\,|\hat{f}_k|^2,
  \label{eq:sobolev_fourier}
\end{equation}
where $\hat{f}_k = \int_0^1 f(t)\,e^{-2\pi i k t}\,dt$ are the Fourier
coefficients of~$f$ \cite[§2.5]{Triebel1983}.

Since $\phi$ is $1$-periodic, the Fourier coefficients of $\phi_m$ are
\begin{equation}
  \widehat{(\phi_m)}_k
  = \begin{cases}
      2^{-ms}\,\hat\phi_{\,k/2^m} & \text{if } 2^m \mid k, \\
      0 & \text{otherwise,}
    \end{cases}
  \label{eq:phi_m_fourier}
\end{equation}
where $\hat\phi_\ell = \int_0^1\phi(t)\,e^{-2\pi i\ell t}\,dt$.
Substituting \eqref{eq:phi_m_fourier} into \eqref{eq:sobolev_fourier}
and changing the summation variable $k = 2^m\ell$
(so that $k$ runs over all multiples of $2^m$):
\begin{equation}
  \sum_{k\in\mathbb{Z}} |k|^{2s'}\,|\widehat{(\phi_m)}_k|^2
  = \sum_{\ell\in\mathbb{Z}} |2^m\ell|^{2s'}\,
    \bigl|2^{-ms}\hat\phi_\ell\bigr|^2
  = 2^{2m(s'-s)}\sum_{\ell\in\mathbb{Z}}|\ell|^{2s'}\,|\hat\phi_\ell|^2
  = 2^{2m(s'-s)}\,\|\phi\|^2_{H^{s'}},
  \label{eq:phi_m_norm_exact}
\end{equation}
where $\|\phi\|^2_{H^{s'}} := \sum_{\ell\in\mathbb{Z}}|\ell|^{2s'}
|\hat\phi_\ell|^2 < \infty$ since $\phi$ is smooth.
For the explicit choice $\phi(t)=\sin^2(\pi t)=\tfrac{1}{2}(1-\cos(2\pi t))$,
the only non-zero coefficients are $\hat\phi_0 = \tfrac{1}{2}$ and
$\hat\phi_{\pm 1} = -\tfrac{1}{4}$,
giving $\|\phi\|^2_{H^{s'}} = 2\cdot\tfrac{1}{16}\cdot 1^{2s'} = \tfrac{1}{8}$
for every $s' \geq 0$, which is finite and explicit.
Note that \eqref{eq:phi_m_norm_exact} is an exact equality, with
no approximation: the change of variable $k=2^m\ell$ is lossless
because $\widehat{(\phi_m)}_k = 0$ whenever $2^m \nmid k$.

Since different terms $\phi_m$ and $\phi_{m'}$ ($m\neq m'$) have
disjoint Fourier supports — for $\phi(t)=\sin^2(\pi t)$, one has
$\widehat{(\phi_m)}_k\neq 0$ only for $k\in\{0,\pm 2^m\}$
(since the only non-zero $\hat\phi_\ell$ are at $\ell\in\{0,\pm 1\}$),
and for $m\neq m'$ with $m,m'\geq 1$ the sets $\{\pm 2^m\}$ and
$\{\pm 2^{m'}\}$ are disjoint — the cross terms in $\|w_s\|^2_{W^{s',2}}$ vanish
identically, giving
\begin{equation}
  \|w_s\|^2_{W^{s',2}(0,1)}
  = \frac{1}{Z_s^2}
    \sum_{m=0}^\infty
    \sum_{k\in\mathbb{Z}}|k|^{2s'}\,|\widehat{(\phi_m)}_k|^2
  = \frac{\|\phi\|^2_{H^{s'}}}{Z_s^2}
    \sum_{m=0}^\infty 2^{2m(s'-s)}.
  \label{eq:ws_sobolev_exact}
\end{equation}
The right-hand side of~\eqref{eq:ws_sobolev_exact} is a geometric
series with ratio $2^{2(s'-s)}$, which converges if and only if
$s'<s$.
Hence $w_s\in W^{s',2}(0,1)$ for every $s'<s$, and
$w_s\notin W^{s,2}(0,1)$.

On each interval $I_i = [i/N,(i+1)/N)$, the function $g_s$ satisfies
$g_s(x) = v_i + (v_{i+1}-v_i)\,w_s(2^n x - i)$.
We compute $|g_s|^2_{W^{s',2}(I_i)}$ directly from the
definition~\eqref{eq:sobolev_seminorm}.
Setting $t = 2^n x - i$ and $u = 2^n y - i$
(so $dx = 2^{-n}\,dt$, $dy = 2^{-n}\,du$, $|x-y| = 2^{-n}|t-u|$):
\begin{align}
  |g_s|^2_{W^{s',2}(I_i)}
  &= \int_0^1\!\int_0^1
     \frac{(v_{i+1}-v_i)^2\,|w_s(t)-w_s(u)|^2}
          {(2^{-n}|t-u|)^{1+2s'}}\,
     2^{-n}\,dt\;2^{-n}\,du \notag\\
  &= (v_{i+1}-v_i)^2\;2^{n(1+2s')}\;2^{-2n}
     \int_0^1\!\int_0^1
     \frac{|w_s(t)-w_s(u)|^2}{|t-u|^{1+2s'}}\,dt\,du \notag\\
  &= (v_{i+1}-v_i)^2\;2^{n(2s'-1)}\;
     |w_s|^2_{W^{s',2}(0,1)}.
  \label{eq:sobolev_cell}
\end{align}
Summing~\eqref{eq:sobolev_cell} over $i = 0,\ldots,N-1$ controls the
within-cell contribution:
\[
  \sum_{i=0}^{N-1}|g_s|^2_{W^{s',2}(I_i)}
  \;\leq\; 2^{n(2s'-1)}\,|w_s|^2_{W^{s',2}(0,1)}
           \sum_{i=0}^{N-1}(v_{i+1}-v_i)^2
  \;<\; \infty,
\]
since $w_s \in W^{s',2}(0,1)$ for every $s'<s$ and the sum over~$i$
is finite.
It remains to control the cross-cell terms.  The function $g_s$ is
bounded by construction, and it has only finitely many jumps, located
at the grid interfaces.  For $s'<s<1/2$, the singular contribution
near each interface is bounded by a constant multiple of
\[
  \int_0^{1/N}\!\int_0^{1/N}\frac{dr\,dt}{(r+t)^{1+2s'}} < \infty,
\]
while cross terms away from interfaces are trivially finite.
Hence $g_s \in W^{s',2}(0,1)$ for every $s'<s$, giving the upper
half of~(ii).

Conversely, take $i^* = 0$. The bit-strings $b(0)=(0,\ldots,0)$ and
$b(1)=(1,0,\ldots,0)$ differ only in bit $k=0$, so
$\Theta(b(1))-\Theta(b(0)) = c_0\neq 0$ by hypothesis.
Both angles lie in $[0,\pi]$, where $\theta\mapsto\sin^2(\theta/2)$ is
strictly increasing; hence $v_1\neq v_0$.
From~\eqref{eq:sobolev_cell} with $s'=s$ and $i^*=0$,
\[
  |g_s|^2_{W^{s,2}(I_0)}
  \;=\; (v_1-v_0)^2\,2^{n(2s-1)}\,
        |w_s|^2_{W^{s,2}(0,1)}.
\]
Since $w_s \notin W^{s,2}(0,1)$, the right-hand side is $+\infty$,
hence $g_s \notin W^{s,2}(0,1)$, giving the lower half of~(ii).
\end{proof}

\begin{remark}[Decoupling of regularity and encoding complexity]
\label{rem:decoupling}
The construction of \cref{lem:gs_construction} shows that the Sobolev
regularity of~$g_s$ and the multilinear degree of~$\Theta_{g_s|_{\mathcal{X}_n}}$ are
\emph{independent} parameters.
The oscillation of $g_s$ — which determines $|g_s|_{W^{s,2}}$ and hence
the classical approximation difficulty — is encoded entirely in the
\emph{within-cell} interpolation via $w_s$.
The multilinear degree of $\Theta_{g_s|_{\mathcal{X}_n}}$, by contrast, depends only on the
\emph{cell boundary values} $\{v_i\}$, which are chosen freely to
satisfy the affine conditions~\eqref{eq:affine_conditions}.
One can therefore make $g_s$ arbitrarily rough (by taking $s\to 0$)
without changing the degree of $\Theta_{g_s|_{\mathcal{X}_n}}$ or the circuit depth of the
encoding operator.
\end{remark}

\section{Computational Illustration}
\label{sec:experiments}

This section provides a computational illustration of the
angle-structure hierarchy $\{\Gclass{n}{d}\}$ at $n=2$ through a
three-tier protocol: a noise-free simulator study, hardware experiments
on the SpinQ Triangulum 3-qubit NMR processor~\cite{SpinQ2021}, and
cross-platform experiments on the IBM Kingston superconducting
processor~\cite{IBMQuantum2024}.
The experiments confirm that the degree stratum $\Gclass{n}{d}$ acts
as a practical hardware-feasibility criterion consistent with the
gate-count formula of \cref{thm:circuit_depth}: circuits at degree
$d=1$ ($n+1=3$ gates) execute successfully on both platforms, while
those at degree $d=2$ ($\binom{2}{\leq 2}=4$ gates) exceed the
Triangulum coherence budget and fail as predicted.

\subsection{Illustrative setup}
\label{subsec:setup}

We estimate $\int_0^1 g(x)\,dx$ for three functions on the 2-qubit
grid $\mathcal{X}_2 = \{0,\tfrac{1}{4},\tfrac{1}{2},\tfrac{3}{4}\}$,
one at each degree level of the hierarchy:
\begin{align}
  g_0(x) &= \tfrac{1}{4},
    && g_0|_{\mathcal{X}_2}\in\Gclass{2}{0},
    && \binom{2}{\leq 0}=1 \text{ gate},
    && \textstyle\int_0^1 g_0 = \tfrac{1}{4},
    \label{eq:g0_def}\\
  g_1(x) &= \sin^2\!\bigl(\tfrac{\pi x}{2}\bigr),
    && g_1|_{\mathcal{X}_2}\in\Gclass{2}{1},
    && \binom{2}{\leq 1}=3 \text{ gates},
    && \textstyle\int_0^1 g_1 = \tfrac{1}{2},
    \label{eq:g1_def}\\
  g_2(x) &= \sin^2(\pi x),
    && g_2|_{\mathcal{X}_2}\in\Gclass{2}{2}\setminus\Gclass{2}{1},
    && \binom{2}{\leq 2}=4 \text{ gates},
    && \textstyle\int_0^1 g_2 = \tfrac{1}{2}.
    \label{eq:g2_def}
\end{align}
The function $g_0$ is the simplest non-trivial case:
a single unconditional rotation $R_Y(\pi/3)$ encodes the
constant amplitude $a=\tfrac{1}{4}$.
It serves as a calibration reference for hardware pulse errors.
The functions $g_1$ and $g_2$ are the benchmark functions of
\cref{ex:sin2,ex:sin2_deg2}.

\begin{remark}[Deliberate choice of test functions]
\label{rem:test_function_design}
The three functions are chosen so that their exact amplitudes
on the $n=2$ midpoint grid, $a_{g_0}=\tfrac{1}{4}$ and
$a_{g_1}=a_{g_2}=\tfrac{1}{2}$, place them in the two
structurally distinct regimes of the MLAE model
$p_k(a)=\sin^2\!\bigl((2k+1)\arcsin\!\sqrt{a}\bigr)$
simultaneously.
\begin{enumerate}[label=(\roman*)]
\item \emph{Degeneracy regime} ($g_0$, $a=\tfrac{1}{4}$).
Here $p_1\!\bigl(\tfrac{1}{4}\bigr)=1$, so the $k=1$ circuit
is saturated at the true amplitude and its Fisher information
$I_1\!\bigl(\tfrac{1}{4}\bigr)$ is undefined
(cf.\ \eqref{eq:fisher_explicit} and the remark following it).
The joint likelihood under $\mathcal{K}=\{0,1\}$ is bimodal
under realistic noise; $\mathcal{K}=\{0,2\}$ resolves the
degeneracy since $p_2\!\bigl(\tfrac{1}{4}\bigr)=\tfrac{1}{4}\neq1$.
\item \emph{Uniform-measurement regime} ($g_1$, $g_2$, $a=\tfrac{1}{2}$).
Here $p_k\!\bigl(\tfrac{1}{2}\bigr)=\tfrac{1}{2}$ for every
$k\geq 0$ (\cref{rem:uniform_k1}), so all circuits produce
equiprobable outcomes regardless of amplification level.
The $k=1$ uniformity observed in the hardware data is therefore
a theoretical property of the schedule, not a decoherence signature.
\end{enumerate}
Together, the three functions probe circuit-depth feasibility
(via the $\Gclass{2}{d}$ stratum), estimation degeneracy
(via $a=\tfrac{1}{4}$), and estimation uniformity
(via $a=\tfrac{1}{2}$).
\end{remark}

All circuits are implemented in SpinQit \cite{SpinQit2022} using the
factorisation of \cref{thm:circuit_depth}; the degree membership is
verified by the classical script
\texttt{00\_check\_function\_affinity.py}~\cite{spinqit_repo}
before each run.
For the IBM Kingston experiments, circuits are built via the Qiskit
SDK~\cite{IBMQuantum2024}, transpiled with \texttt{optimization\_level=0}
to prevent the default gate-fusion pass from altering the CCRy
decomposition, and submitted via the \texttt{SamplerV2} primitive
of the Qiskit Runtime API.
Both hardware campaigns use the midpoint quadrature rule,
schedule $\mathcal{K}=\{0,1\}$ (and $\{0,1,2\}$ for $g_0$),
and $N=2048$ shots per level.
Raw job data and post-processing scripts are available in
\cite{spinqit_repo}.
Explicit circuit diagrams are given in \cref{fig:circuits,app:circuit_proof}.

\subsection{Simulator results}
\label{subsec:sim_validation}

\Cref{tab:results_sim} reports MLAE estimates from the noise-free
SpinQit simulator.
All nine entries achieve absolute error below $10^{-7}$, confirming
\cref{thm:circuit_depth,thm:pareto} to machine precision.

\begin{table}[h]
\centering
\caption{MLAE estimates on the SpinQit simulator for
$g_0|_{\mathcal{X}_2}\in\Gclass{2}{0}$ ($d=0$, 1 gate),
$g_1|_{\mathcal{X}_2}\in\Gclass{2}{1}$ ($d=1$, 3 gates), and
$g_2|_{\mathcal{X}_2}\in\Gclass{2}{2}\setminus\Gclass{2}{1}$ ($d=2$, 4 gates).
Schedule $\mathcal{K}=\{0,1\}$, 2048 shots per level.
Exact values: $\tfrac{1}{4}$ for $g_0$, $\tfrac{1}{2}$ for $g_1, g_2$.}
\label{tab:results_sim}
\begin{threeparttable}
\renewcommand{\arraystretch}{1.3}
\begin{tabular}{llcc}
\toprule
Function & Rule & $\hat{I}$ & Abs.\ error \\
\midrule
$g_0$ ($d=0$, 1 gate)
  & Left Riemann  & $0.2499999974$ & $<10^{-7}$ \\
  & Midpoint      & $0.2499999974$ & $<10^{-7}$ \\
  & Right Riemann & $0.2499999974$ & $<10^{-7}$ \\
  & Simpson       & $0.2499999974$ & $<10^{-7}$ \\
\midrule
$g_1$ ($d=1$, 3 gates)
  & Left Riemann  & $0.3749999974$ & $1.25\times10^{-1}$\tnote{a} \\
  & Midpoint      & $0.4999999987$ & $<10^{-7}$\tnote{b} \\
  & Right Riemann & $0.6249999976$ & $1.25\times10^{-1}$\tnote{a} \\
  & Simpson       & $0.4999999983$ & $<10^{-7}$\tnote{b} \\
\midrule
$g_2$ ($d=2$, 4 gates)
  & Left Riemann  & $0.4999999987$ & $<10^{-7}$\tnote{b} \\
  & Midpoint      & $0.4999999987$ & $<10^{-7}$\tnote{b} \\
  & Right Riemann & $0.4999999987$ & $<10^{-7}$\tnote{b} \\
  & Simpson       & $0.4999999987$ & $<10^{-7}$\tnote{b} \\
\bottomrule
\end{tabular}
\begin{tablenotes}
\footnotesize
\item[a] The exact discretisation error is
  $|I[g_1] - R_2^{\mathrm{left}}[g_1]|
   = \bigl|\tfrac{1}{2} - \tfrac{1}{4}(g_1(0)+g_1(\tfrac{1}{4})
     +g_1(\tfrac{1}{2}))\bigr|
   = \bigl|\tfrac{1}{2} - \tfrac{1}{4}(0+\tfrac{1}{2}+1)\bigr|
   = \tfrac{1}{8} = 1.25\times10^{-1}$,
  matching the observed value exactly.
  This is consistent with the bound \eqref{eq:error_left}:
  $E_2^{\mathrm{left}} \leq \tfrac{1}{2}\|g_1'\|_\infty/4
   = \tfrac{\pi}{32}\approx 0.098$;
  the bound is not tight because $\|g_1'\|_\infty=\pi/4$ is not
  attained at the left-endpoint error quadrature points.
  The right Riemann error is identical by symmetry of $g_1$
  on $\mathcal{X}_2$.
\item[b] Machine-precision accuracy: $p_k\bigl(\tfrac{1}{2}\bigr)
  = \tfrac{1}{2}$ for all $k\geq 0$ when $a=\tfrac{1}{2}$ exactly;
  see \cref{rem:uniform_k1}.
  For $g_2$, this holds for all rules because the discrete sum
  $\tfrac{1}{4}\sum_i g_2(x_i)=\tfrac{1}{2}$ exactly by symmetry
  of $\sin^2(\pi x)$ on the $n=2$ grid.
\end{tablenotes}
\end{threeparttable}
\end{table}

\begin{remark}[Uniform $k=1$ measurements at $a=\tfrac{1}{2}$]
\label{rem:uniform_k1}
When the true amplitude is $a=\tfrac{1}{2}$, the MLAE model gives
\begin{equation}
  p_k\!\bigl(\tfrac{1}{2}\bigr)
  = \sin^2\!\Bigl((2k+1)\arcsin\!\bigl(\sqrt{\tfrac{1}{2}}\bigr)\Bigr)
  = \sin^2\!\bigl((2k+1)\tfrac{\pi}{4}\bigr)
  = \tfrac{1}{2}
  \label{eq:uniform_at_half}
\end{equation}
for every $k\geq 0$.
Both the $k=0$ and $k=1$ circuits then produce uniform measurement
distributions, and the MLE converges to $\hat\theta=\pi/4$ from the
$k=0$ data alone.
This is a theoretical property of the schedule $\mathcal{K}=\{0,1\}$
at $a=\tfrac{1}{2}$, not a decoherence signature.
The same phenomenon appears in the hardware results for $g_1$ and $g_2$
(\cref{tab:pk_kingston,tab:results_g1_hw}), where it must be distinguished from
the decoherence-induced uniformity observed for $g_0$ and discussed in
\cref{subsec:hw_discussion}.
\end{remark}

\subsection{Hardware illustration}
\label{subsec:hw_results}

We validate the theoretical predictions on two hardware platforms
that differ substantially in scale, technology, and coherence budget,
providing complementary perspectives on the practical relevance of
the $\Gclass{n}{d}$ hierarchy.
The \emph{SpinQ Triangulum} is a 3-qubit NMR processor
\cite{SpinQ2021} representing a low-cost, laboratory-scale device
with a tight line-depth limit of~60.
The \emph{IBM~Kingston} is a 127-qubit superconducting processor
\cite{IBMQuantum2024} accessed via the Qiskit Runtime cloud API,
representing an industrial-scale device without an operative depth
constraint for the circuits studied here.
Both devices run the same three test functions
$g_0|_{\mathcal{X}_2}\in\Gclass{2}{0}$, $g_1|_{\mathcal{X}_2}\in\Gclass{2}{1}$,
$g_2|_{\mathcal{X}_2}\in\Gclass{2}{2}\setminus\Gclass{2}{1}$
with schedule $\mathcal{K}=\{0,1,2\}$ and 2048 shots per level,
using the midpoint quadrature rule throughout.

\paragraph{Triangulum: proof of concept on minimal hardware.}
On the Triangulum, $g_0$ and $g_1$ compile within the
line-depth limit; $g_2$ does not.
The 4-gate encoding of $g_2$ requires the doubly-controlled
gate $C^{\{0,1\}}R_Y(-\pi)$, which after NMR compilation by
SpinQit \cite{SpinQit2022} exceeds 60 lines for all values of~$k$.
This is a direct consequence of \cref{thm:circuit_depth}: the
$\mathcal{O}(n^2)$ gate count of $\Gclass{n}{2}$ circuits places
them beyond the coherence budget of the Triangulum, while the
$\mathcal{O}(n)$ gate count of $\Gclass{n}{1}$ does not.
The hierarchy thus manifests on this device as a sharp feasibility
threshold rather than a gradual accuracy degradation.

\Cref{tab:results_g0_hw,tab:results_g1_hw} report MLAE estimates
for $g_0$ and $g_1$ on the Triangulum.
For $g_1|_{\mathcal{X}_2}\in\Gclass{2}{1}$, the midpoint rule achieves
$|\hat{a}-\tfrac{1}{2}|=8.5\times10^{-3}$ (mean over 6 runs,
std $0.003$), confirming stable execution at the edge of the
coherence budget.
For $g_0|_{\mathcal{X}_2}\in\Gclass{2}{0}$, the MLAE schedule $\mathcal{K}=\{0,1\}$
produces bimodal estimates ($\hat{I}\approx0.13$ in 4 runs,
$\hat{I}\approx0.39$ in 2 runs), reflecting the structural
identifiability failure $I_1(\tfrac{1}{4})=0$ discussed in
\cref{subsec:hw_discussion}.

\begin{table}[h]
\centering
\caption{MLAE estimates for $g_0(x)=\tfrac{1}{4}$, $g_0|_{\mathcal{X}_2}\in\Gclass{2}{0}$
  ($d=0$, 1 gate) on SpinQ Triangulum, $\mathcal{K}=\{0,1\}$,
  2048 shots per level. Exact value: $\tfrac{1}{4}$.}
\label{tab:results_g0_hw}
\begin{threeparttable}
\renewcommand{\arraystretch}{1.3}
\begin{tabular}{lccc}
\toprule
Rule & $\hat{I}$ (mean) & Std & Abs.\ error \\
\midrule
Left Riemann    & $0.1329274$ & ---     & $1.17\times10^{-1}$\tnote{c} \\
Midpoint        & $0.2168730$ & $0.133$\tnote{a} & $3.31\times10^{-2}$ \\
Right Riemann   & $0.1326396$ & ---     & $1.17\times10^{-1}$\tnote{c} \\
Simpson         & $0.1888432$ & ---     & $6.12\times10^{-2}$\tnote{b} \\
\bottomrule
\end{tabular}
\begin{tablenotes}
\footnotesize
\item[a] Mean and std over 6 independent runs.
  4 runs give $\hat{I}\approx 0.13$; 2 give $\hat{I}\approx 0.39$
  (bimodal MLE, see \cref{subsec:hw_discussion}).
  The $k=0$ estimator $\hat{a}_{k=0}=\hat{p}_0$ is more robust:
  mean $0.2345\pm0.018$, corresponding to NMR pulse offset
  $\Delta\theta\approx-0.036$\,rad; see \cref{rem:nmr_calibration}.
\item[b] Computed as $(L+4M+R)/6$ using MLE midpoint mean.
\item[c] Left and right Riemann rules are single runs each.
  The midpoint rule was repeated 6 times because it is the primary
  benchmark quadrature and the MLE bimodality required statistical
  characterisation (\cref{subsec:hw_discussion}).
  Hardware access time on the Triangulum limited the total number
  of runs available.
\end{tablenotes}
\end{threeparttable}
\end{table}

\begin{table}[h]
\centering
\caption{MLAE estimates for
  $g_1(x)=\sin^2\!\bigl(\tfrac{\pi x}{2}\bigr)$, $g_1|_{\mathcal{X}_2}\in\Gclass{2}{1}$
  ($d=1$, 3 gates) on SpinQ Triangulum, $\mathcal{K}=\{0,1\}$,
  2048 shots per level. Exact value: $\tfrac{1}{2}$.}
\label{tab:results_g1_hw}
\begin{threeparttable}
\renewcommand{\arraystretch}{1.3}
\begin{tabular}{lccc}
\toprule
Rule & $\hat{I}$ (mean) & Std & Abs.\ error \\
\midrule
Left Riemann    & $0.4804039$ & ---     & $1.96\times10^{-2}$ \\
Midpoint        & $0.5085047$ & $0.003$ & $8.50\times10^{-3}$ \\
Right Riemann   & $0.5162186$ & ---     & $1.62\times10^{-2}$ \\
Simpson         & $0.5051069$ & ---     & $5.11\times10^{-3}$\tnote{a} \\
\bottomrule
\end{tabular}
\begin{tablenotes}
\footnotesize
\item[a] Computed as $(L+4M+R)/6$ using MLE midpoint mean.
  Midpoint and right Riemann give $\hat{p}_{k=1}=\tfrac{1}{2}$
  exactly — theoretical consequence of $a=\tfrac{1}{2}$
  (\cref{rem:uniform_k1}): when $a=\tfrac{1}{2}$, the model
  probability $p_k(\tfrac{1}{2})=\tfrac{1}{2}$ for all $k\geq 0$
  regardless of noise level.
  This is distinguishable from decoherence-induced uniformity
  (as observed for $g_0$, \cref{subsec:hw_discussion}) because the
  independent $k=0$ estimate $\hat{p}_0$ confirms $a\approx\tfrac{1}{2}$
  to within hardware noise, ruling out a decoherence origin.
\end{tablenotes}
\end{threeparttable}
\end{table}

\paragraph{Kingston: full validation across the hierarchy.}
On the IBM Kingston, all three functions are hardware-executable.
\Cref{tab:results_hw_kingston} reports MLAE estimates for the
optimal schedule of each function.
The per-level probabilities $\hat{p}_k$ collected in
\cref{tab:pk_kingston} confirm quantitative agreement with
the theoretical model $p_k(a)=\sin^2\bigl((2k+1)\arcsin(\sqrt{a})\bigr)$
at the measured noise levels.

\begin{table}[h]
\centering
\caption{MLAE estimates on IBM Kingston for
  $g_0|_{\mathcal{X}_2}\in\Gclass{2}{0}$ ($d=0$, 1~gate),
  $g_1|_{\mathcal{X}_2}\in\Gclass{2}{1}$ ($d=1$, 3~gates), and
  $g_2|_{\mathcal{X}_2}\in\Gclass{2}{2}\setminus\Gclass{2}{1}$ ($d=2$, 4~gates).
  Midpoint quadrature, 2048 shots per level.
  Exact values: $a=\tfrac{1}{4}$ for $g_0$, $a=\tfrac{1}{2}$
  for $g_1$ and $g_2$.}
\label{tab:results_hw_kingston}
\renewcommand{\arraystretch}{1.3}
\begin{threeparttable}
\begin{tabular}{llccc}
\toprule
Function & Schedule & $\hat{a}$ & $|\hat{a}-a|$ & $a$ \\
\midrule
$g_0$ ($d=0$, 1~gate)  & $\mathcal{K}=\{0,2\}$\tnote{a}
  & $0.25178$ & $1.78\times10^{-3}$ & $\tfrac{1}{4}$ \\
$g_1$ ($d=1$, 3~gates) & $\mathcal{K}=\{0,1\}$
  & $0.50093$ & $9.28\times10^{-4}$ & $\tfrac{1}{2}$ \\
$g_2$ ($d=2$, 4~gates) & $\mathcal{K}=\{0,1\}$
  & $0.50127$ & $1.27\times10^{-3}$ & $\tfrac{1}{2}$ \\
\bottomrule
\end{tabular}
\begin{tablenotes}
\footnotesize
\item[a] $\mathcal{K}=\{0,2\}$ skips the $k=1$ level, where
  $p_1(\tfrac{1}{4})=1$ exactly and the Fisher information $I_1(\tfrac{1}{4})$
  is undefined (cf.\ \eqref{eq:fisher_explicit} and
  \cref{prop:admissible_schedule,ex:degeneracy_g0});
  see \cref{subsec:hw_discussion}.
  The na\"{i}ve schedule $\mathcal{K}=\{0,1\}$ yields
  $|\hat{a}-\tfrac{1}{4}|=8.4\times10^{-2}$, 48~times worse.
\end{tablenotes}
\end{threeparttable}
\end{table}

\begin{table}[h]
\centering
\caption{Observed ancilla probabilities $\hat{p}_k$ on IBM Kingston,
  2048 shots per level, midpoint rule.
  Theoretical value: $p_k^{\mathrm{th}}(a)=\sin^2\bigl((2k+1)\arcsin(\sqrt{a})\bigr)$.}
\label{tab:pk_kingston}
\renewcommand{\arraystretch}{1.3}
\begin{tabular}{lcccccc}
\toprule
& \multicolumn{2}{c}{$g_0$\;($a=\tfrac{1}{4}$)}
& \multicolumn{2}{c}{$g_1$\;($a=\tfrac{1}{2}$)}
& \multicolumn{2}{c}{$g_2$\;($a=\tfrac{1}{2}$)} \\
\cmidrule(lr){2-3}\cmidrule(lr){4-5}\cmidrule(lr){6-7}
$k$ & $\hat{p}_k$ & $p_k^{\mathrm{th}}$
    & $\hat{p}_k$ & $p_k^{\mathrm{th}}$
    & $\hat{p}_k$ & $p_k^{\mathrm{th}}$ \\
\midrule
$0$ & $0.2520$ & $0.2500$ & $0.4756$ & $0.5000$ & $0.5039$ & $0.5000$ \\
$1$ & $0.9087$ & $1.0000$ & $0.4888$ & $0.5000$ & $0.4971$ & $0.5000$ \\
$2$ & $0.2412$ & $0.2500$ & $0.4580$ & $0.5000$ & $0.5146$ & $0.5000$ \\
\bottomrule
\end{tabular}
\end{table}

\subsection{Discussion}
\label{subsec:hw_discussion}

\paragraph{$g_0$: Fisher information degeneracy.}
The constant function $g_0=\tfrac{1}{4}$ has $a=\tfrac{1}{4}$,
so $p_1(\tfrac{1}{4})=\sin^2\bigl(3\arcsin(\tfrac{1}{2})\bigr)
=\sin^2(\tfrac{\pi}{2})=1$ exactly.
The Fisher information $I_1\!\bigl(\tfrac{1}{4}\bigr)$ is undefined
at the true parameter value: although the denominator $a(1-a)$ of
\eqref{eq:fisher_explicit} is nonzero at $a=\tfrac{1}{4}$, the formula
requires $p_k(a)\notin\{0,1\}$, and here $p_1(\tfrac{1}{4})=1$ violates
that condition.
Including $k=1$ in the MLAE schedule therefore adds noise without
adding information.
On the Triangulum, this manifests as bimodal MLE under
$\mathcal{K}=\{0,1\}$: hardware noise attenuates $\hat{p}_1$
below~1, creating two nearly equiprobable likelihood maxima
(4~runs give $\hat{I}\approx0.13$, 2~runs give $\hat{I}\approx0.39$).
On Kingston, the attenuation is smaller ($\hat{p}_1=0.909$ vs
$p_1^{\mathrm{th}}=1$, corresponding to a depolarising noise level
$\varepsilon\approx18\%$ at the $k=1$ circuit depth), but
$\mathcal{K}=\{0,1\}$ still gives $|\hat{a}-\tfrac{1}{4}|=8.4\times10^{-2}$,
48~times worse than $\mathcal{K}=\{0,2\}$, which correctly bypasses
the degenerate level.
The prescription is therefore device-independent: for any~$a$ such
that $p_k(a)=1$ for some~$k$, that level should be excluded from
the MLAE schedule \cite{SuzukiEtAl2019}.
Under $\mathcal{K}=\{0,2\}$, Kingston achieves
$|\hat{a}-\tfrac{1}{4}|=1.78\times10^{-3}$ from a single
2048-shot run, consistent with the $\mathcal{O}(1/\sqrt{N})$
classical baseline at $k=0$.

\paragraph{$g_1$: cross-platform validation of the $d=1$ circuit.}
For $g_1|_{\mathcal{X}_2}\in\Gclass{2}{1}$ with $a=\tfrac{1}{2}$, the model gives
$p_k(\tfrac{1}{2})=\tfrac{1}{2}$ for all~$k$ (\cref{rem:uniform_k1}),
so depolarising noise at any level leaves $\hat{p}_k$ invariant
around~$\tfrac{1}{2}$.
The measured values on Kingston ($\hat{p}_0=0.476$,
$\hat{p}_1=0.489$, $\hat{p}_2=0.458$) are all within $2.2\sigma$
of the theoretical value, confirming that the $k=0,1$ circuits
operate well within the coherence budget of the superconducting
device.
Under $\mathcal{K}=\{0,1\}$, Kingston achieves
$|\hat{a}-\tfrac{1}{2}|=9.3\times10^{-4}$, roughly $9\times$
better than the Triangulum result of $8.5\times10^{-3}$ at the
same shot count.
The improvement reflects the lower depolarising noise of the
superconducting gate set rather than a difference in the
encoding circuit, confirming that the $\Gclass{2}{1}$ encoding
is portable across platforms and that hardware fidelity is the
binding constraint at this circuit depth.

\paragraph{$g_2$: hierarchy confirmation across platforms.}
The degree-2 function $g_2=\sin^2(\pi x)$ requires the
doubly-controlled gate $C^{\{0,1\}}R_Y(-\pi)$, the single gate
that distinguishes $\Gclass{2}{2}$ from $\Gclass{2}{1}$ in
\cref{thm:circuit_depth}.
On the Triangulum, its inclusion pushes all circuits beyond the
line-depth limit of~60, making $g_2$ entirely infeasible.
On Kingston, the same gate transpiles correctly within the native
ECR gate set, and all three $\hat{p}_k$ values for $g_2$
lie within $1.3\sigma$ of $\tfrac{1}{2}$, indistinguishable from
ideal behaviour.
Under $\mathcal{K}=\{0,1\}$, Kingston achieves
$|\hat{a}-\tfrac{1}{2}|=1.27\times10^{-3}$, comparable to the
$g_1$ result at the same shot count ($9.28\times10^{-4}$).
Taken together, the two platforms illustrate that the
$\Gclass{n}{d}$ stratum is not merely a theoretical complexity
parameter but a practical hardware feasibility criterion whose
threshold depends on the coherence budget of the target device:
$d=1$ is universally feasible at $n=2$, while $d=2$ requires
a device whose gate fidelity accommodates the additional
controlled rotation.

\paragraph{Comparison with prior work.}
The IBM~Q Ourense experiments of \cite{CarreraVazquezWoerner2020}
used $k_{\max}=16$ and up to 8192~shots to reach midpoint errors of
order $10^{-4}$ for~$g_1$.
With $k_{\max}=1$ and 2048~shots, Kingston reaches
$9.3\times10^{-4}$ --- within one order of magnitude at $64\times$
fewer shots --- and the Triangulum reaches $8.5\times10^{-3}$.
The $\Gclass{2}{1}$ encoding therefore remains competitive under
shallow schedules, and viable even on a severely depth-constrained
NMR device.

\begin{remark}[NMR pulse calibration errors]
\label{rem:nmr_calibration}
On the Triangulum, the $k=0$ estimator $\hat{p}_0\approx0.2345$
(mean over 6 runs, std $0.018$) for $g_0$ deviates from the
theoretical $a=\tfrac{1}{4}$, yielding an effective rotation angle
$\hat\theta_{\mathrm{eff}}=2\arcsin(\sqrt{0.2345})\approx1.011$\,rad
against the nominal $\theta=\pi/3\approx1.047$\,rad,
a pulse-length offset $\Delta\theta\approx-0.036$\,rad ($-2.1^\circ$).
This is a well-known systematic in NMR quantum computing
\cite{VandersypenChuang2005} that does not affect $g_1$,
because the dominant encoding rotations for the midpoint rule
lie near $\theta=\pi/2$ where
$\tfrac{d}{d\theta}\sin^2(\theta/2)=\tfrac{1}{2}\sin\theta$ vanishes.
On Kingston, no analogous offset is observed for $g_0$:
$\hat{p}_0=0.252$, within $0.2\sigma$ of $a=\tfrac{1}{4}$.
Resolving the Triangulum offset requires vendor-level
recalibration of the pulse parameters.
\end{remark}

\section{Conclusions and Open Problems}
\label{sec:conclusions}

We have introduced the angle-structure hierarchy $\{\Gclass{n}{d}\}$ as
a framework for characterising the encoding-circuit complexity of state
preparation for QAE-based numerical integration.
The multilinear degree of the angle map $\Theta_{g|_{\mathcal{X}_n}}:\mathbb{B}^n\to[0,\pi]$
is a classically computable invariant (via the Walsh--Hadamard transform,
in $\mathcal{O}(n2^n)$ time) that controls the gate count of the
canonical monomial-factorisation circuit: $g|_{\mathcal{X}_n}\in\Gclass{n}{d}$
implies an encoding with exactly $\binom{n}{\leq d}$ controlled-$\Ry$
gates (an upper bound, tight at $d=n$ by \cite{ShendeBullockMarkov2006}).

The main theoretical results are:
\begin{itemize}
  \item the \emph{encoding-circuit theorem} (\cref{thm:circuit_depth}),
        giving the gate count $\binom{n}{\leq d}$ of the monomial
        factorisation circuit, which is an upper bound on encoding
        complexity and is tight at $d=n$;
  \item the \emph{depth-accuracy trade-off} (\cref{thm:pareto}),
        showing that the total gate count to achieve $\varepsilon$-accuracy
        is $\mathcal{O}((\log1/\varepsilon)^d/\varepsilon)$.  At the
        same discretisation level, the corresponding classical Monte
        Carlo sample count scales as
        $\mathcal{O}(\varepsilon^{-(2+1/p)})$
        (cf.\ \cref{rem:scope} for the scope of this comparison);
  \item the \emph{decoupling of encoding degree from regularity}
        (\cref{lem:gs_construction}), constructing functions
        $g_s:[0,1]\to[0,1]$ with
        $g_s|_{\mathcal{X}_n}\in\Gclass{n}{1}$ of arbitrarily low
        Sobolev regularity, whose encoding depth is $\mathcal{O}(n)$
        irrespective of~$s$; the angle-structure degree is therefore
        not a smoothness invariant.
\end{itemize}

The computational illustration on SpinQ Triangulum and IBM Kingston
(\cref{sec:experiments}) confirms that the $\Gclass{n}{d}$ degree
stratum is a practical hardware-feasibility criterion at $n=2$: the
feasibility threshold is device-dependent (governed by the coherence
budget), while the gate-count formula of \cref{thm:circuit_depth} is
platform-independent.
To the authors' knowledge this is the first cross-platform experimental
comparison of QAE circuits stratified by the angle-structure degree,
though we do not claim a comprehensive survey of the QAE hardware
literature.
Extensions to multivariate functions and Markov-chain transition kernels
follow directly from \cref{cor:multivariate}: the gate-count formula
$\binom{n}{\leq d}$ depends only on the total qubit count $n=\sum_k n_k$
and the degree of $\Theta_{g|_{\mathcal{X}_n}}$ on $\mathbb{B}^n$,
regardless of how the bits are partitioned across dimensions or time steps;
a detailed treatment of the Heston model and related applications is
identified as a direction for future work.

\paragraph{Open problems.}

\begin{enumerate}[label=(\roman*)]
\item
  \textbf{Tightness of the gate-count bound.}
  \cref{thm:circuit_depth} gives an upper bound of $\binom{n}{\leq d}$
  controlled-$\Ry$ gates for the encoding of $g\in\Gclass{n}{d}$;
  this bound is tight at $d=n$ by \cite{ShendeBullockMarkov2006}.
  Is $\binom{n}{\leq d}$ the minimum gate count for $d < n$, or can
  some functions in $\Gclass{n}{d}$ be encoded more efficiently using
  ancilla qubits or structured decompositions?
  Resolving this question would determine whether "encoding-circuit
  complexity" in the sense of this paper captures the true gate
  complexity or only an efficiently constructible upper bound.

\item
  \textbf{Optimal degree for a given $g$.}
  Given $g\in C^\alpha[0,1]$, what is the minimum $d^*$ such that
  $g|_{\mathcal{X}_n}\in\Gclass{n}{d^*}$, and how
  does $d^*$ depend on~$n$ and~$\alpha$?

\item
  \textbf{Approximation within $\Gclass{n}{d}$.}
  For $g\notin\Gclass{n}{d}$, what is the best approximation
  $\tilde{g}\in\Gclass{n}{d}$ in the $L^2(\mathcal{X}_n)$ sense?
  Truncating the multilinear expansion of $\Thetag$ to degree~$d$ is a
  natural candidate; bounding the resulting integration error in terms
  of the tail $\sum_{|S|>d}|\wh{\Theta}{S}|$ would yield a practical
  degree-selection criterion.

\item
  \textbf{Heston model and Markov kernels.}
  \cref{cor:multivariate} establishes that the hierarchy extends verbatim
  to multivariate functions and Markov-chain transition kernels.
  A concrete open problem is the full Möbius analysis of the
  \emph{Heston stochastic volatility model}
  \cite{CarreraVazquezWoerner2020}: the linearised discretisation places
  the transition kernel in $\Gclass{4}{2}\setminus\Gclass{4}{1}$ due to
  the variance-price correlation, but a self-contained treatment
  — including the explicit SDE system, the discretisation scheme, the
  Möbius coefficient table, and numerical values for standard Heston
  parameters — is left for future work.

\item
  \textbf{Extension to non-uniform grids.}
  The hierarchy is defined for uniform grids; extending to Gaussian
  quadrature nodes would yield higher-order convergence at the cost
  of a more complex angle map.

\item
  \textbf{Noise robustness.}
  Under realistic noise models, what is the minimum degree~$d$ such
  that the QAE sampling gain remains visible, and how does the
  depth-versus-accuracy trade-off curve deform in the presence of noise?
\end{enumerate}

\subsection*{Acknowledgements}

The authors thank the SpinQ Technology team for technical support with
the Triangulum device and access to the SpinQit framework,
and acknowledge the use of IBM Quantum services for the Kingston
experiments.
This research was funded by the grant number COMCUANTICA/007 from the Generalitat Valenciana, 
Spain, by the grant number INDI24/17 from the Universidad CEU Cardenal Herrera, Spain.

\subsection*{Conflicts of interest}
The authors declare no conflicts of interest.

\subsection*{Declaration of generative AI and AI-assisted technologies in the writing process}
During the preparation of this work the authors used Claude (Anthropic)
to improve the readability and language of the manuscript.
After using this tool/service, the authors reviewed and edited the
content as needed and take full responsibility for the content of the
published article.

\appendix
\section{Quantum Computing Model}
\label{app:qc}

This appendix develops the quantum computing model in standard
linear-algebra language, consistent with \cref{subsec:qprob}.
\Cref{app:dirac} translates every object into Dirac bra-ket notation.
A detailed treatment is given in \cite{FalcoEtAl2025,FalcoMatthies2026}.

\subsection{Quantum circuits as unitary maps}

An $n$-qubit quantum computer operates on $V_n = \mathbb{C}^{2^n}$
with standard orthonormal basis $\{\mathbf{u}_k\}_{k=0}^{2^n-1}$.
Each basis vector factors as
\[
  \mathbf{u}_k = \mathbf{u}_{b_0(k)}^{(0)} \otimes \cdots
                 \otimes \mathbf{u}_{b_{n-1}(k)}^{(n-1)},
\]
where $b_j(k)\in\{0,1\}$ is the $j$-th bit of $k$ in the binary
expansion~\eqref{eq:index_map}, and $\mathbf{u}_0^{(\ell)},
\mathbf{u}_1^{(\ell)}$ denote the standard basis of the $\ell$-th
qubit factor $\mathbb{C}^2_{(\ell)}$.
A \emph{quantum circuit} is a product of elementary unitary matrices
whose composition yields $U \in \mathrm{U}(2^n)$.
The gate dictionary used throughout is $\{R_Y(\theta), X, \mathrm{CNOT}\}$
where
\[
R_Y(\theta) = \begin{pmatrix}
\cos\frac{\theta}{2} & -\sin\frac{\theta}{2} \\
\sin\frac{\theta}{2} & \phantom{-}\cos\frac{\theta}{2}
\end{pmatrix} \in \mathrm{U}(2),
\qquad
X = \begin{pmatrix} 0 & 1 \\ 1 & 0 \end{pmatrix} \in \mathrm{U}(2),
\]
and $\mathrm{CNOT}(c\to t)$ acts on the two-qubit factor
$\mathbb{C}^2_{(c)}\otimes\mathbb{C}^2_{(t)}$ of $V_n$ by
\begin{equation*}
  \mathrm{CNOT}(c\to t)\,(\mathbf{u}_a^{(c)} \otimes \mathbf{u}_b^{(t)})
  = \mathbf{u}_a^{(c)} \otimes \mathbf{u}_{a \oplus b}^{(t)},
  \qquad a,b \in \{0,1\},
\end{equation*}
where $\oplus$ denotes addition modulo~$2$: the gate flips qubit~$t$
if and only if qubit~$c$ is in state~$\mathbf{u}_1^{(c)}$, and acts
as the identity on all remaining qubits $\ell \notin \{c,t\}$.
In the ordered basis
$\{\mathbf{u}_0^{(c)}\otimes\mathbf{u}_0^{(t)},\,
   \mathbf{u}_0^{(c)}\otimes\mathbf{u}_1^{(t)},\,
   \mathbf{u}_1^{(c)}\otimes\mathbf{u}_0^{(t)},\,
   \mathbf{u}_1^{(c)}\otimes\mathbf{u}_1^{(t)}\}$
of $\mathbb{C}^2_{(c)}\otimes\mathbb{C}^2_{(t)}$
this action yields the matrix
\[
\mathrm{CNOT}(c\to t) = \begin{pmatrix}
1 & 0 & 0 & 0 \\
0 & 1 & 0 & 0 \\
0 & 0 & 0 & 1 \\
0 & 0 & 1 & 0
\end{pmatrix}.
\]

Although $R_Y(\theta)$ and $X$ act on $\mathbb{C}^2$, they extend to
$\mathrm{U}(2^n)$ by tensoring with the identity: for a gate
$U\in\mathrm{U}(2)$ acting on qubit~$i$, the corresponding $n$-qubit
operator is
\[
  U_{(i)} := I_2^{\otimes i} \otimes U \otimes I_2^{\otimes(n-i-1)}
  \;\in\;\mathrm{U}(2^n),
  \qquad i\in\{0,\ldots,n-1\}.
\]
For example, $R_Y(\theta)_{(2)}$ applies the rotation $R_Y(\theta)$ to
qubit~$2$ while leaving qubits $0,\ldots,1$ unchanged.
Similarly, $\mathrm{CNOT}(c\to t)$ acts non-trivially only on the
two-qubit factor $\mathbb{C}^2\otimes\mathbb{C}^2$ indexed by
$(c,t)$, with identity on the remaining $n-2$ qubits.
Starting from the initial vector $\mathbf{e}_0 = \mathbf{u}_0$,
the circuit maps $\mathbf{e}_0 \mapsto U\mathbf{e}_0$.

\subsection{Quantum random variables and measurement}

The projectors $\Pi_k = \mathbf{u}_k\mathbf{u}_k^*$ introduced in
\cref{subsec:qprob} are the measurement operators.
For a circuit $U \in \mathrm{U}(2^n)$ starting from $\mathbf{e}_0$,
measuring in the basis $\{\mathbf{u}_k\}$ yields outcome~$k$ with
probability~\eqref{eq:circuit_prob}.
For circuits on $V_n \otimes \mathbb{C}^2$ (index register plus
ancilla), the marginal probability of ancilla outcome~$1$ is given
by~\eqref{eq:marginal_linalg}; this is the central quantity
estimated throughout the paper.

\subsection{The amplitude oracle and Grover iterate}

The encoding operator $G_g$, the amplitude oracle
$A_g = G_g(H_n\otimes I_2)$ and the Grover iterate
$Q = A_g S_{\mathbf{e}_0} A_g^* S_{\mathbf{w}_0}$ were defined in
\cref{subsec:oracle}, in the same linear-algebra notation used here;
we record below only the spectral facts needed for the circuit
constructions of \cref{app:circuit_proof}.
On the two-dimensional invariant subspace spanned by
$\{A_g(\mathbf{e}_0\otimes\mathbf{u}_0^{(2)}),
   A_g(\mathbf{e}_0\otimes\mathbf{u}_1^{(2)})\}$,
$Q$ acts as a rotation by $2\arcsin(\sqrt{a})$,
yielding \eqref{eq:qae_prob} \cite{BrassardHoyerMoscaTapp2002}.

\subsection{Spin-Echo circuit optimisation}

Each evaluation of $Q^kA_g$ nominally requires $2k+1$ applications
of the encoding layer $G_g$.
The identity \cite{CarreraVazquezWoerner2020}
\begin{equation}
  R_Y(-f)\,(I_n\otimes Z)\,R_Y(f) = R_Y(-2f)\,(I_n\otimes Z)
  \label{eq:spinecho_id}
\end{equation}
allows one application of $G_g$ per Grover iterate to be absorbed,
reducing the encoding count from $2k+1$ to $k+1$.
This is the Spin-Echo optimisation used in \cref{sec:experiments}.

\subsection{Dictionary: linear-algebra and Dirac notation}
\label{app:dirac}

The Dirac bra-ket notation common in the quantum computing literature
is a typographic convention on top of the same linear-algebra structures.
The following table gives the complete correspondence.

\begin{center}
\renewcommand{\arraystretch}{1.5}
\begin{tabular}{lll}
\toprule
Linear-algebra object & Dirac notation & Reading \\
\midrule
$\mathbf{u}_k \in \mathbb{C}^N$ & $\ket{k}$ & ``ket $k$'' (column vector) \\
$\mathbf{u}_k^* \in (\mathbb{C}^N)^*$ & $\bra{k}$ & ``bra $k$'' (row vector / covector) \\
$\langle \mathbf{u}_j, \mathbf{v}\rangle = \mathbf{u}_j^*\mathbf{v}$ & $\braket{j}{\psi}$ & inner product \\
$\Pi_k = \mathbf{u}_k\mathbf{u}_k^*$ & $\ket{k}\bra{k}$ & rank-one projector \\
$\mathbf{v} = \sum_k\alpha_k\mathbf{u}_k$ & $\ket{\psi}=\sum_k\alpha_k\ket{k}$ & state expansion \\
$U\mathbf{e}_0$ & $U\ket{0}^{\otimes n}$ & circuit output \\
$\mathbf{u}_i\otimes\mathbf{u}_b^{(2)}$ & $\ket{i}\otimes\ket{b}$ & product basis vector \\
$I_N\otimes\mathbf{u}_1^{(2)}(\mathbf{u}_1^{(2)})^*$ & $I_n\otimes\ket{1}\bra{1}$ & ancilla projector $\Pi_1^{\mathrm{anc}}$ \\
$\mathbf{v}^*A\mathbf{v}$ & $\bra{\psi}A\ket{\psi}$ & expectation of $A$ in state $\mathbf{v}$ \\
$S_{\mathbf{e}_0} = I - 2\mathbf{e}_0\mathbf{e}_0^*$ & $I - 2\ket{0}\bra{0}^{\otimes n+1}$ & Grover reflection \\
\bottomrule
\end{tabular}
\end{center}

The only structural difference is that Dirac notation makes the
primal/dual distinction typographic — $\ket{\cdot}$ for column
vectors, $\bra{\cdot}$ for row vectors — so that the pairing
$\braket{j}{\psi} = \mathbf{u}_j^*\mathbf{v}$ is immediate by
juxtaposition.
Every identity in this paper translates word-for-word between the
two notations via the table above.

\section{Circuit Implementation Details}
\label{app:circuit_proof}

\subsection{The affine case $\Gclass{n}{1}$}

For $g:\mathcal{X}_n\to[0,1]$ with $g\in\Gclass{n}{1}$, the angle map is affine:
$\Theta_g(b)=\hat\Theta_\emptyset+\sum_{j=0}^{n-1}\hat\Theta_{\{j\}}b_j$.
By \cref{thm:circuit_depth}, the encoding operator factors as
\begin{equation}
  G_g = \Ry(\hat\Theta_\emptyset)\cdot
        \prod_{j=0}^{n-1}C^{\{j\}}\Ry(\hat\Theta_{\{j\}}),
  \label{eq:Gg_affine}
\end{equation}
using $n+1$ gates: one unconditional rotation and $n$ single-controlled
rotations.
Although all $n+1$ gates commute (each controlled by a different index
qubit and all acting on the same ancilla), they cannot be executed
simultaneously because they share the ancilla qubit as their target.
Each gate therefore occupies a separate layer, giving circuit depth
$n+1 = \mathcal{O}(n)$.

\subsection{The quadratic case $\Gclass{n}{2}$}

For $g:\mathcal{X}_n\to[0,1]$ with $g\in\Gclass{n}{2}$, the expansion adds $\binom{n}{2}$
two-qubit-controlled terms:
\begin{equation}
  G_g = \Ry(\hat\Theta_\emptyset)\cdot
        \prod_j C^{\{j\}}\Ry(\hat\Theta_{\{j\}})\cdot
        \prod_{j<k}C^{\{j,k\}}\Ry(\hat\Theta_{\{j,k\}}),
  \label{eq:Gg_quadratic}
\end{equation}
for a total of $\binom{n}{\leq 2}=1+n+\binom{n}{2}$ gates.
Each $C^{\{j,k\}}\Ry(\alpha)$ decomposes as
$\cnot(k\to\mathrm{anc})\cdot C_j$-$\Ry(\alpha/2)\cdot
\cnot(k\to\mathrm{anc})\cdot C_j$-$\Ry(-\alpha/2)$
\cite{BarencoEtAl1995}, costing 2 CNOTs per pair,
where $\cnot(k\to\mathrm{anc})$ denotes the CNOT with control qubit~$k$
and target the ancilla qubit.

\paragraph{Example: $n=2$, $g_2=\sin^2(\pi x)$, $g_2|_{\mathcal{X}_2}\in\Gclass{2}{2}$.}
From \cref{ex:sin2_deg2}: $\hat\Theta_\emptyset=0$,
$\hat\Theta_{\{0\}}=\tfrac{\pi}{2}$,
$\hat\Theta_{\{1\}}=\pi$,
$\hat\Theta_{\{0,1\}}=-\pi$.
The gate sequence is:
(1) $C_{q_0}$-$\Ry(\tfrac{\pi}{2})$,
(2) $C_{q_1}$-$\Ry(\pi)$,
(3) $C^{\{q_0,q_1\}}$-$\Ry(-\pi)$ (the degree-2 correction,
    implemented as 2 CNOTs + 2 controlled-$\Ry$).
Total: 4 controlled-$\Ry$ gates $=\binom{2}{\leq 2}$, one more than
the 3 gates for $g_1|_{\mathcal{X}_2}\in\Gclass{2}{1}$.

\begin{figure}[t]
\centering
\includegraphics[width=\textwidth]{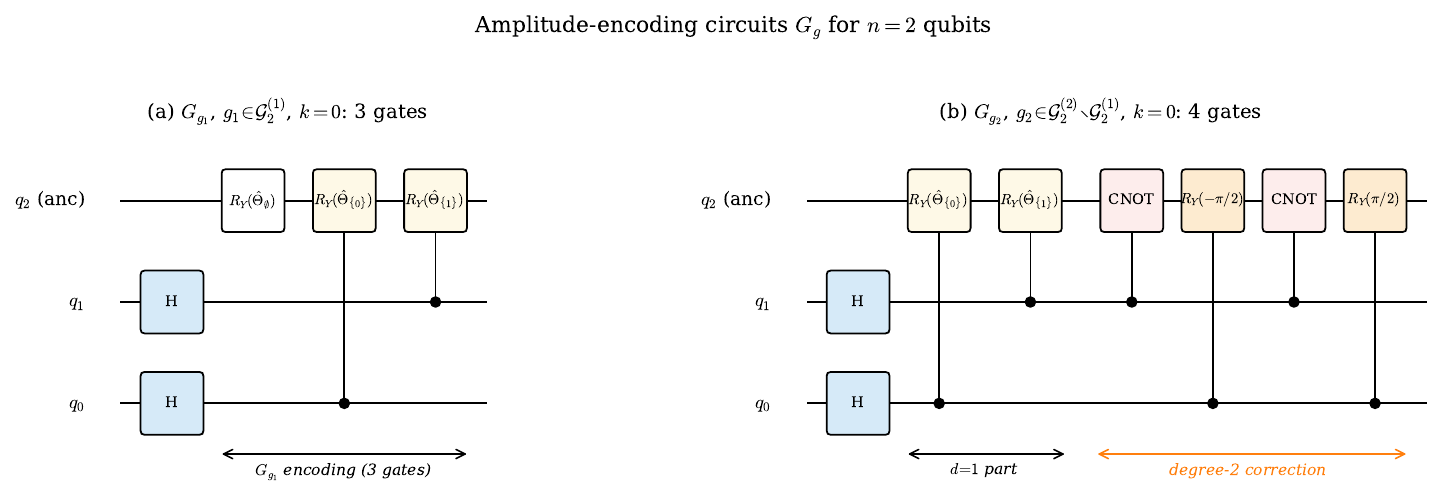}
\caption{Amplitude-encoding circuits $G_g$ for $n=2$ qubits, $k=0$.
  \textbf{(a)} $g_1|_{\mathcal{X}_2}\in\Gclass{2}{1}$ ($d=1$, 3 gates, \eqref{eq:Gg_affine}):
  one unconditional rotation $R_Y(\hat\Theta_\emptyset)$ on the ancilla
  $q_2$, one $C_{q_0}$-$R_Y(\hat\Theta_{\{0\}})$, and one
  $C_{q_1}$-$R_Y(\hat\Theta_{\{1\}})$.
  Circuit depth $\leq 4$ layers.
  \textbf{(b)} $g_2|_{\mathcal{X}_2}\in\Gclass{2}{2}\setminus\Gclass{2}{1}$ ($d=2$, 4 gates,
  \eqref{eq:Gg_quadratic}): the $d=1$ part (left brace) is followed by the
  degree-2 correction $C^{\{0,1\}}R_Y(\hat\Theta_{\{0,1\}})$ (orange brace),
  decomposed into 2 CNOTs and 2 singly-controlled rotations
  \cite{BarencoEtAl1995}.
  Hadamard gates $H$ on $q_0, q_1$ prepare the uniform superposition
  for the amplitude oracle $A_g = G_g(H_2\otimes I_2)$.}
\label{fig:circuits}
\end{figure}

\paragraph{Compiled circuit depths for hardware runs.}
For $g_1|_{\mathcal{X}_2}\in\Gclass{2}{1}$, $k=0$: Hadamard gates on both index qubits
followed by the 3-gate sequence of \eqref{eq:Gg_affine}.
Compiled depth $\leq 4$ layers on the Triangulum — well within
the line-depth limit of~60.
For $g_1$, $k=1$: the Spin-Echo optimisation of \cref{app:qc} reduces
the number of $G_g$ applications from 3 to 2, but the resulting
compiled depth still approaches or exceeds the Triangulum coherence
budget, causing the observed decoherence signature
(\cref{sec:experiments}).
On IBM Kingston the same $g_1$ circuits run without depth constraints:
there is no operative line-depth limit, and both $k=0$ and $k=1$
execute faithfully (\cref{tab:pk_kingston}).
For $g_2|_{\mathcal{X}_2}\in\Gclass{2}{2}$, $k=0$: the additional
$C^{\{0,1\}}R_Y(-\pi)$ gate raises the compiled depth to
$\approx 8$--$10$ layers (below the Triangulum limit in principle),
but the $k=1$ circuit exceeds~60 after full NMR pulse compilation by
SpinQit \cite{SpinQit2022}, making the full MLAE schedule infeasible
on the Triangulum.
On IBM Kingston, all three levels $k=0,1,2$ of the $g_2$ circuit
execute successfully, with all $\hat{p}_k$ within shot noise of the
theoretical value; see \cref{tab:pk_kingston,subsec:hw_results}.

\subsection{SpinQ Triangulum device specifications}

\begin{center}
\renewcommand{\arraystretch}{1.3}
\begin{tabular}{ll}
\toprule
Parameter & Value \\
\midrule
Number of qubits & 3 ($q_0, q_1, q_2$) \\
Connectivity & All-to-all \\
Native gates & $\Ry(\theta)$, $X$, CNOT \\
Line-depth limit & 60 \\
Coherence time $T_2$ & $\sim10$--$30$\,ms (qubit-dependent) \\
Gate time (single-qubit) & $\sim1$\,ms \\
Gate time (CNOT) & $\sim2$--$5$\,ms \\
Software & SpinQit \cite{SpinQit2022} \\
\bottomrule
\end{tabular}
\end{center}

Qubit assignment: $q_0=b_0$ (index LSB), $q_1=b_1$ (index MSB),
$q_2=\text{ancilla}$.
Code at \cite{spinqit_repo}.

\paragraph{Coherence budget for $g_2$ at $k=1$.}
The encoding of $g_2|_{\mathcal{X}_2}\in\Gclass{2}{2}$ requires the gate
$C^{\{0,1\}}R_Y(-\pi)$, which decomposes into 2 CNOTs and
2 controlled-$R_Y$ rotations (\cref{app:circuit_proof}).
For the Grover iterate at $k=1$ with Spin-Echo, the full circuit
contains approximately $2(n+1)+2=8$ multi-qubit operations.
At $\sim2$\,ms per CNOT (optimistic estimate from the table above),
this already approaches the lower end of the coherence time
$T_2\sim10$\,ms.
After SpinQit NMR compilation, which expands multi-qubit gates into
native shaped-pulse sequences, the resulting line-depth exceeds the
hardware limit of~60, confirming hardware infeasibility for $g_2$
at $k=1$.
At $k=0$ the compiled depth is $\approx8$--$10$ layers
(below the limit in principle), and future reductions in
NMR compilation overhead could enable the $d=2$ case.

\Cref{tab:depth_summary} collects the gate counts, compiled depths,
and hardware feasibility for all three test functions.

\begin{table}[h]
\centering
\caption{Gate count, compiled line-depth, and hardware feasibility
  for $n=2$ qubits, schedule $\mathcal{K}=\{0,1,2\}$, 2048 shots
  per level.
  Triangulum depth refers to compiled line-depth after NMR pulse
  expansion by SpinQit; feasibility threshold is 60 lines.
  IBM Kingston has no operative depth limit for these circuits.}
\label{tab:depth_summary}
\begin{threeparttable}
\renewcommand{\arraystretch}{1.3}
\begin{tabular}{lccccc}
\toprule
Function & Class & Gates & Triangulum & IBM Kingston & Runs \\
\midrule
$g_0=\tfrac{1}{4}$
  & $\Gclass{2}{0}$ & 1
  & Feasible\tnote{a} & Feasible
  & Tri: 6\;/\;King: 3 \\
$g_1=\sin^2\!\tfrac{\pi x}{2}$
  & $\Gclass{2}{1}$ & 3
  & $k{=}0$ only\tnote{b} & Feasible
  & Tri: 6\;/\;King: 3 \\
$g_2=\sin^2(\pi x)$
  & $\Gclass{2}{2}$ & 4
  & Infeasible\tnote{c} & Feasible
  & Tri: ---\;/\;King: 3 \\
\bottomrule
\end{tabular}
\begin{tablenotes}
\footnotesize
\item[a] Triangulum $g_0$: midpoint repeated 6 times for statistical
  characterisation of MLE bimodality
  (\cref{subsec:hw_discussion}); 1 run each for left/right Riemann.
\item[b] Triangulum $g_1$: $k=1$ compiled depth approaches or
  exceeds the coherence budget $T_2\sim10$\,ms after Spin-Echo
  optimisation; decoherence signature observed
  (\cref{subsec:hw_discussion}).
  Kingston $g_1$: all three levels $k=0,1,2$ execute faithfully
  (\cref{tab:pk_kingston}).
\item[c] Triangulum $g_2$: $C^{\{0,1\}}R_Y(-\pi)$ pushes the
  compiled line-depth beyond 60 for all $k$.
  Kingston $g_2$: all three levels execute successfully;
  $\hat{p}_k$ within $1.3\sigma$ of $\tfrac{1}{2}$
  (\cref{tab:pk_kingston}).
\end{tablenotes}
\end{threeparttable}
\end{table}

\section{Walsh--Hadamard Transform on \texorpdfstring{$\mathbb{B}^n$}{B\^{}n}}
\label{app:walsh}

This appendix collects background on multilinear polynomials and the
Walsh--Hadamard transform.
For a comprehensive reference see \cite{O'Donnell2014}.

\subsection{Multilinear coefficients}

Every $f:\mathbb{B}^n\to\R$ has the unique multilinear expansion
\eqref{eq:fourier_expansion}, with coefficients given by Möbius
inversion \eqref{eq:mobius}.
The degree $\deg(f)=\max\{|S|:\hat{f}_S\neq 0\}$, and
$\deg(\Theta_{g|_{\mathcal{X}_n}})\leq d$ is equivalent to the vanishing conditions
\eqref{eq:membership_conditions}.

\subsection{Walsh--Hadamard transform}

Identify $\mathbb{B}^n$ with $\{0,\ldots,2^n-1\}$ via $i(b)$ and
form $\mathbf{f}=(f(b))_{b\in\mathbb{B}^n}\in\R^{2^n}$.
The WHT computes $\hat{\mathbf{f}}=H_n\mathbf{f}$, where
$H_n=H^{\otimes n}$ and
$H=\frac{1}{\sqrt{2}}\bigl(\begin{smallmatrix}1&1\\1&-1\end{smallmatrix}\bigr)$.
(The unnormalised version, omitting $1/\sqrt{2}$, gives the Möbius
coefficients up to a global sign.)
The WHT is computed in-place in $n$ passes of $2^{n-1}$ butterfly
operations each, for total cost $\mathcal{O}(n2^n)$ — the same
asymptotic complexity as the Fast Fourier Transform.

\subsection{Relationship to Boolean Fourier analysis}

In Boolean Fourier analysis \cite{O'Donnell2014}, functions on
$\{-1,+1\}^n$ are expanded in the parity basis $\chi_S(x)=\prod_{j\in S}x_j$.
Under $b_j\leftrightarrow(1-x_j)/2$, the monomial $\chimon{S}(b)$ maps
to a combination of parities of degree $\leq|S|$, so the two notions
of degree coincide.
We use the $\{0,1\}$-encoding because it matches the computational
basis and makes \eqref{eq:mobius} more direct.

\begin{remark}[Relation to the cryptographic Walsh-support literature]
\label{rem:walsh_support_crypto}
The $\mathbb{F}_2$-valued counterpart of this framework — in which one
studies the support of the Walsh transform of a Boolean function
$f:\mathbb{F}_2^n\to\mathbb{F}_2$ as a subset of $\mathbb{F}_2^n$ — has
been investigated extensively in the cryptographic literature;
see \cite{CarletMesnager2004,CarletPastorTornero2025}.
In that setting the analogue of the condition $\deg(\Theta_g)\leq d$ is
the requirement that the Walsh support be contained in a union of flats
of bounded dimension, and a set $S\subset\mathbb{F}_2^n$ is called
\emph{fully balanced} if its Fourier--Hadamard transform takes only the
values $\{-|S|,0,|S|\}$; by \cite[Theorem~1]{CarletPastorTornero2025}
this is equivalent to $S$ being an affine space, which corresponds to
the degenerate case $d=0$ of our hierarchy.
The present paper works in the real-valued regime
$\Theta_g:\mathbb{B}^n\to[0,\pi]$, where the relevant invariant is the
multilinear degree rather than the spectral support structure over
$\mathbb{F}_2^n$, and the application is circuit-complexity analysis for
QAE-based numerical integration rather than cryptographic distinguishing
of function classes.
\end{remark}

\bibliographystyle{plain}

\end{document}